\documentclass[letterpaper,10pt,conference]{ieeeconf}
\IEEEoverridecommandlockouts
\usepackage{amsmath,amssymb,amsfonts}
\usepackage{mathtools}
\usepackage{graphicx}
\usepackage{xcolor}
\usepackage{url}
\usepackage{cite}
\usepackage{booktabs}
\usepackage{array}
\usepackage{balance}
\usepackage{algorithm}
\usepackage{hyperref}
\usepackage{algpseudocode}
\usepackage{subcaption}
\usepackage{todonotes}
\usepackage{tikz}
\usetikzlibrary{arrows.meta,positioning,calc}

\newcommand{\BibTeX}{\rm B\kern-.05em{\sc i\kern-.025em b}\kern-.08em\TeX}
\newtheorem{theorem}{Theorem}
\newtheorem{lemma}[theorem]{Lemma}
\newtheorem{proposition}[theorem]{Proposition}

\newtheorem{definition}[theorem]{Definition}
\newtheorem{remark}[theorem]{Remark}

\title{Distributed Learning of $\alpha$-Potential Functions in Strategic Multi-agent Systems}
\author{Yifei Chen$^{1}$ and Chinmay Maheshwari$^{2}$
\thanks{$^{1}$Yifei Chen is with the Department of Applied Mathematics and Statistics,
        Johns Hopkins University, Baltimore, Maryland, USA
        {\tt\small (email: ychen780@jhu.edu)}}%
\thanks{$^{2}$Chinmay Maheshwari is with the Department of Electrical and Computer Engineering, Johns Hopkins University, Baltimore, Maryland, USA
        {\tt\small (email: chinmay\_maheshwari@jhu.edu)}}%
}

\begin{document}
\maketitle

\begin{abstract}
We study the problem of computing the tightest \(\alpha\)-potential
approximation of a general-sum game over continuous action spaces,
within a prescribed class of potential functions and when each player
has access only to its own utility function. The difficulty is twofold:
the approximation error involves a worst-case search over an infinite
set of unilateral deviations, and the required utility information is
distributed across players.

For a linear-in-parameters potential class, we use an exact finite-tuple
reformulation that separates the problem into a global outer search
over deviation tuples and distributed convex inner problems. We develop
a primal--dual inner oracle tailored to this structure and establish a
uniform one-sided accuracy guarantee. This oracle can be combined with
global outer search to obtain an end-to-end guarantee on the outer
optimization error. We also develop a projected zeroth-order outer
method as a computationally lighter alternative for higher-dimensional
problems. Numerical experiments illustrate the accuracy--computation
tradeoff between the two outer-search methods and show that the proposed
optimization framework can improve upon analytical \(\alpha\)-potential
constructions.
\end{abstract}
\section{Introduction}
Strategic interactions with heterogeneous and coupled objectives arise
in networked control, resource allocation, and multi-agent learning
\cite{marden2009cooperative,li2022confluence}. Potential games simplify
the analysis of such systems by representing players' unilateral
utility changes through a common scalar function
\cite{MONDERER1996124}. Exact potential functions, however, need not
exist in general-sum games. Near-potential games instead quantify the
mismatch through an approximation parameter $\alpha$; smaller $\alpha$
yields tighter equilibrium and learning guarantees for the original
game~\cite{candogan2011near,guo2025markov,rokade2025asymmetric}.
This motivates computing the smallest approximation error achievable
within a prescribed class of potential functions.

For finite-action games, unilateral deviations can be enumerated and
the nearest potential can be computed by linear programming
\cite{candogan2013dynamics}. The continuous-action setting is
substantially more difficult: the deviation space is infinite, and
evaluating the uniform full-domain approximation error requires a
generally nonconcave global maximization over deviations, even when
the objective is convex in the potential parameters. Moreover, in
many multi-agent systems utility information is naturally distributed:
each player can evaluate its own utility but need not know the utility
functions of other players. A centralized computation would therefore
require collecting player-specific utility information. We instead
seek a distributed procedure that keeps these utility evaluations
local and coordinates the computation through communication among
players.

Existing continuous-action approaches typically exploit additional
game structure rather than directly optimizing the full-domain error.
For example, explicit or analytical $\alpha$-potential constructions
have been developed for asymmetric network games
\cite{rokade2025asymmetric}, dynamic congestion games
\cite{guo2025markov}, and dynamic $N$-player games
\cite{guo2024alpha}, among others. Such constructions provide useful
analytical bounds but need not yield the tightest approximation for a
particular game instance. NePPO~\cite{kalanther2026neppo} instead
optimizes an equilibrium-local proxy rather than the uniform
full-domain criterion. Motivated by this gap, we study the distributed
computation of the tightest full-domain approximation over continuous
actions when each player has access only to its own utility function.

We consider a potential class that is linear in its parameters, while
allowing nonlinear dependence on the players' actions through a
prescribed feature map. Using the finite-tuple reformulation of
\cite{EXOTIC}, we express the original problem as a global outer
maximization over finite collections of deviations, where evaluating
each candidate tuple requires solving a convex optimization problem
over the potential parameters. This yields a modular two-level
architecture: the inner problem can be solved using distributed
optimization with local utility evaluations, while different search
procedures can be used for the outer problem.

For the inner problem, general distributed constrained-optimization
methods could in principle be applied
\cite{SrivastavaNedicStipanovic,aybat2016primal,xu2026accelerated}.
We instead develop an oracle tailored to the finite-maximum structure
of our problem. Each player maintains a single scalar constraint
multiplier, evaluates only its local utility, and communicates its
primal iterate with neighboring players. The resulting network-averaged
output provides a one-sided estimate of the finite-tuple value with an
error bound that is uniform over all queried tuples. This uniform
oracle property allows us to combine the distributed inner solver with global tree search method
EXOTIC \cite{EXOTIC} and obtain an end-to-end guarantee for the outer objective.
We also consider projected zeroth-order ascent as a computationally
lighter alternative when hierarchical global search becomes expensive.
Our main contributions are as follows.

\noindent\textbf{(1)}
We formulate the distributed computation of the tightest
near-potential approximation, within a prescribed
linear-in-parameters potential class, for continuous-action
general-sum games. Applying the finite-tuple reformulation of
\cite{EXOTIC}, we obtain an equivalent outer maximization over
deviation tuples with convex inner problems. We also give sufficient
conditions under which an optimizer of the finite-tuple problem
recovers an optimizer of the original full-domain problem (Proposition \ref{prop:reformulation-recovery}).

\noindent\textbf{(2)}
We develop a distributed primal--dual oracle tailored to these inner
problems. By retaining their finite-maximum structure, each player
uses one scalar constraint multiplier, accesses only its local utility,
and exchanges only its primal iterate with neighboring players. For
an inner iteration budget \(K\), we establish a uniform one-sided
oracle error of order
\(\widetilde{\mathcal O}(K^{-1/4})\) over all deviation tuples (Theorem \ref{thm:drpdm-convergence}).

\noindent\textbf{(3)}
We combine the distributed oracle with EXOTIC and derive an end-to-end convergence guarantee for the optimal near-potential value (Theorem \ref{thm:exotic-global}). Together with the structural recovery conditions of Proposition \ref{prop:reformulation-recovery}, this further yields asymptotic recovery of an optimal potential parameter.
We additionally develop a projected zeroth-order outer method as a
scalable heuristic that avoids hierarchical tree search. Numerical
experiments compare the resulting accuracy--computation tradeoff and
show that the optimization-based approach can improve upon analytical
near-potential constructions.

\paragraph{Notation}
For a scalar $\lambda$, let
$[\lambda]_+:=\max\{\lambda,0\}$. Given a nonempty closed convex set
$\Omega$, let $P_\Omega(x)$ denote the Euclidean projection of $x$
onto $\Omega$. For a convex function $f$, $\partial f(x)$ denotes
its sub-differential at $x$, while $\nabla_x f$ denotes its gradient
with respect to $x$ whenever it exists.
The communication network is represented by a connected undirected graph
$\mathcal G=(\mathcal V_\mathcal G,\mathcal E_\mathcal G)$, where $\mathcal V_\mathcal G$ is identical to the player set $\mathcal I$ and $\mathcal E_\mathcal G\subseteq \mathcal V_\mathcal G\times \mathcal V_\mathcal G$ is the edge set. The neighbor set of player $i$ is
$\mathcal N_i:=\{j\in\mathcal I:(i,j)\in \mathcal E_\mathcal G\}.$
Let $A_{\mathcal G}$ be the adjacency matrix, with
$[A_{\mathcal G}]_{ij}=1$ if $(i,j)\in\mathcal E_{\mathcal G}$
and zero otherwise. Let
$D_{\mathcal G}:=\operatorname{diag}(d_1,\ldots,d_N)$, where
$d_i:=\sum_j[A_{\mathcal G}]_{ij}$, and define the graph Laplacian
as $L_{\mathcal G}:=D_{\mathcal G}-A_{\mathcal G}$.
  
\section{Problem Formulation}
\label{sec:problem-formulation}
Consider a strategic-form game
$\Gamma=(\mathcal I,\{\mathcal A_i\}_{i\in\mathcal I},
\{u_i\}_{i\in\mathcal I})$
with player set $\mathcal I:=\{1,\ldots,N\}$,
nonempty compact action sets $\mathcal A_i\subseteq\mathbb R^{d_i}$,
and Lipschitz continuous utility functions
$u_i:\mathcal A\to\mathbb R$, where
$\mathcal A:=\prod_{i\in\mathcal I}\mathcal A_i$.
For each $i\in\mathcal I$, write $a=(a_i,a_{-i})$,
where $a_{-i}$ denotes the actions of all other players.
Let $\mathcal Y:=\mathcal A^2$.
For $y=(a,a')\in\mathcal Y$ and $q:\mathcal A\to\mathbb R$,
define the unilateral difference
$\Delta_q(y;i):=q(a_i',a_{-i})-q(a_i,a_{-i})$.

\begin{definition}[Exact potential function \cite{MONDERER1996124}]
A function $\phi:\mathcal A\to\mathbb R$ is an exact potential
function for $\Gamma$ if
$\Delta_\phi(y;i)=\Delta_{u_i}(y;i)$
for all $i\in\mathcal I$ and $y\in\mathcal Y$.
A game admitting such a function is an exact potential game.
\end{definition}
General-sum games need not admit an exact potential function.
We therefore consider approximate potentials that uniformly
bound the discrepancy between unilateral potential and utility
differences.
\begin{definition}[$\alpha$-potential function]
For $\alpha\geq 0$, a function $\phi:\mathcal A\to\mathbb R$
is an $\alpha$-potential function for $\Gamma$ if
$|\Delta_\phi(y;i)-\Delta_{u_i}(y;i)|\leq\alpha$
for all $i\in\mathcal I$ and $y\in\mathcal Y$.
\end{definition}

The case $\alpha=0$ recovers an exact potential function.
Our objective is to find the smallest $\alpha$ achievable
within a prescribed class of candidate potentials. We consider
\begin{equation}\label{eq:pot_param}
\phi_\theta(a):=\langle\theta,\psi(a)\rangle,
\qquad \theta\in\Theta\subseteq\mathbb R^d,
\end{equation}
where $\psi:\mathcal A\to\mathbb R^d$ is a prescribed Lipschitz
continuous feature map and $\Theta$ is nonempty, compact,
and convex\footnote{Although we use linear parametrization of potential function in \eqref{eq:pot_param}, the results in this paper extend to convex parameterization setting.}.

For each $i\in\mathcal I$ and $y\in\mathcal Y$, define
the signed approximation error \(
\ell_i(\theta,y)
:=\Delta_{\phi_\theta}(y;i)-\Delta_{u_i}(y;i),\)
and let
\begin{equation}
\label{eq:local-global-objectives}
f(\theta,y):=\max_{i\in\mathcal I}\ell_i(\theta,y),
\qquad
\mathcal E(\theta):=\max_{y\in\mathcal Y}f(\theta,y).
\end{equation}
Since $\mathcal Y$ includes both directions of every unilateral
deviation, $\mathcal E(\theta)
=\max_{i\in\mathcal I,\,y\in\mathcal Y}|\ell_i(\theta,y)|$.
Thus, $\mathcal E(\theta)$ is the smallest $\alpha$ for which
$\phi_\theta$ is an $\alpha$-potential function.
The resulting optimization problem is
\begin{equation}
\label{eq:general-potential-learning}
p^\star:=\min_{\theta\in\Theta}\mathcal E(\theta)
=\min_{\theta\in\Theta}\max_{y\in\mathcal Y}f(\theta,y),
\end{equation}
with an optimal parameter denoted by
$\theta^\star\in\arg\min_{\theta\in\Theta}\mathcal E(\theta)$. 

Computing the optimal approximation in
\eqref{eq:general-potential-learning} presents two main
challenges. First, although the objective is convex in
the potential parameters, its evaluation requires solving
a generally nonconcave maximization over the continuous
space of unilateral deviations. Second, evaluating the
objective requires information about all players'
utilities, whereas each player has access only to
its own utility function.

We address these challenges in the next two sections.
Section \ref{sec:finite-tuple} derives an equivalent finite-tuple reformulation
that separates the global nonconvex search from convex
optimization over the potential parameters.
Section \ref{sec:distributed-optimization} develops a distributed algorithm for solving
these convex subproblems using local utility information
and combines it with outer optimization methods.

\begin{remark}[Choice of parameterization]
Linearity in \(\theta\) makes each \(\ell_i(\cdot,y)\) affine
and hence \(f(\cdot,y)\) and \(\mathcal E\) convex, while
expressiveness is determined by \(\psi\).
For a finite action space
\(\mathcal A=\{a^1,\ldots,a^M\}\), choosing \(d=M\) and
\(\psi(a^m)=e_m\) recovers any tabular potential permitted by
\(\Theta\). On compact continuous action spaces, expanding a
polynomial basis yields uniform approximation of continuous
potentials by the Stone--Weierstrass
theorem~\cite{stone1948generalized}. The features may instead be
fixed neural-network representations or discounted state-action
occupancy measures, as in~\cite[Section~5]{guo2024alpha}.
Thus, linearity in the parameters still permits nonlinear dependence
on actions or policies. Nevertheless, \(f\) remains generally
nonconcave in \(y\), making~\eqref{eq:general-potential-learning}
convex--nonconcave minimax problem.
\end{remark}

\section{Finite-Tuple Reformulation}
\label{sec:finite-tuple}
We begin by reformulating the convex-nonconcave minimax optimization problem
\eqref{eq:general-potential-learning} to nonconcave convex maximin optimization problem using results in~\cite{EXOTIC}. The key idea is to replace
the maximization over the continuous deviation space
with an outer search over finite collections of deviations.
For each such collection, the corresponding optimization
over the potential parameters is convex.

Specifically, let $\mathcal W:=\mathcal Y^{d+1}$ denote
the set of tuples containing $d+1$ deviations.
For each $w=(y_1,\ldots,y_{d+1})\in\mathcal W$, define
\begin{equation}
\label{eq:finite-tuple-functions}
F(\theta,w)
:=
\max_{1\leq j\leq d+1} f(\theta,y_j),
\qquad
G(w)
:=
\min_{\theta\in\Theta}F(\theta,w).
\end{equation}

For a fixed tuple $w$, the function $F(\cdot,w)$
measures the maximum approximation error over the
deviations contained in $w$. Since $F(\cdot,w)$
is convex, computing $G(w)$ amounts to solving
a convex optimization problem over the potential
parameters.
Moreover, because
$F(\theta,w)\leq\mathcal E(\theta)$ for every
$\theta\in\Theta$ and $w\in\mathcal W$, we have
$G(w)\leq p^\star$.
Thus, every tuple provides a lower bound on the
optimal approximation error.

The following result establishes that maximizing
these lower bounds over all tuples recovers the
optimal value of the original problem. It also
provides sufficient conditions under which an
optimal parameter of the finite-tuple problem
solves the original problem.

\begin{proposition}[Reformulation and optimizer recovery]
\label{prop:reformulation-recovery}

Under the standing assumptions,
\begin{equation}
\label{eq:outer-optimal-value}
p^\star
=
\max_{w\in\mathcal W}G(w).
\end{equation}

Moreover, let \(
w^\star
=
(y_1^\star,\ldots,y_{d+1}^\star)
\in
\operatorname*{arg\,max}_{w\in\mathcal W}G(w),\)
and \(
\theta_{w^\star}^\star
\in
\operatorname*{arg\,min}_{\theta\in\Theta}
F(\theta,w^\star).\)
Define the set of active player-deviation pairs
\[
S^\star
:=
\left\{
(i,j)\in\mathcal I\times[d+1]:
\ell_i(\theta_{w^\star}^\star,y_j^\star)
=
F(\theta_{w^\star}^\star,w^\star)
\right\}.
\]

Then $\theta_{w^\star}^\star$ is an optimal
solution of~\eqref{eq:general-potential-learning}
if either of the following conditions holds:

\begin{itemize}
    \item[(i)]
    The finite-tuple objective is exact at
    $\theta_{w^\star}^\star$, i.e., \(
    F(\theta_{w^\star}^\star,w^\star)
    =
    \mathcal E(\theta_{w^\star}^\star).\)

    \item[(ii)]
    The active feature-difference vectors
    positively span $\mathbb R^d$, i.e., \(
    0\in
    \operatorname{int}\operatorname{conv}
    \left\{
    \Delta_\psi(y_j^\star;i):
    (i,j)\in S^\star
    \right\}.\)
\end{itemize}
\end{proposition}

\begin{proof}[Proof Sketch]
The optimal-value identity
\eqref{eq:outer-optimal-value}
follows from~\cite[Proposition~2.1]{EXOTIC}.
The required continuity and convexity
conditions hold, while compactness of $\Theta$
ensures boundedness and attainment.

For optimizer recovery, we use \cite[Proposition 2.4]{EXOTIC}. Condition~(i)
directly implies that
$\mathcal E(\theta_{w^\star}^\star)=p^\star$ (\cite[Proposition 2.4-(ii)]{EXOTIC}).
Under condition~(ii), the active
feature-difference vectors ensure that
$\theta_{w^\star}^\star$ is the unique minimizer
of $F(\cdot,w^\star)$. Since every minimizer
of $\mathcal E$ also minimizes
$F(\cdot,w^\star)$, it follows that
$\theta_{w^\star}^\star$ is also the
unique minimizer of $\mathcal E$(\cite[Proposition 2.4-(i)]{EXOTIC}).
The complete proof is provided in
Appendix~\ref{app:reformulation-recovery-proof}.
\end{proof}
Proposition~\ref{prop:reformulation-recovery} separates value recovery from parameter recovery. Equality \eqref{eq:outer-optimal-value} implies that solving the outer problem recovers the optimal approximation value \(p^\star\)  without additional assumptions. 
For an optimal tuple $w^\star$, an inner minimizer solves the
full-domain problem if it satisfies either condition~(i) or
condition~(ii) of Proposition~\ref{prop:reformulation-recovery}. Section~\ref{sec:distributed-optimization} establishes
convergence of the outer objective values to $p^\star$.
This value convergence alone does not guarantee recovery of
a full-domain optimizer; the optimizer-recovery statement
remains conditional on the conditions in Proposition~\ref{prop:reformulation-recovery}.

\section{Distributed Oracle and Outer Optimization}
\label{sec:distributed-optimization}
Proposition~\ref{prop:reformulation-recovery} motivates a
modular two-level computational framework, illustrated
in Figure~\ref{fig:outer-inner}. The framework separates
the global maximization over deviation tuples from the
convex optimization required to evaluate $G(w)$.

In this paper, we develop a distributed inner oracle
tailored to the finite-tuple problem structure.
At each outer query, a centralized coordinator
broadcasts a candidate tuple $w$ and an iteration
budget $K$. Players then use local utility information
and communication with their neighbors to compute
an approximate evaluation $\widehat G_K(w)$.
We establish a uniform finite-time accuracy guarantee
for the proposed oracle.

For the outer optimization, we consider two methods
that offer different computational trade-offs.
First, we integrate our distributed oracle with
EXOTIC~\cite{EXOTIC}, which provides a global
optimization guarantee but can be computationally
expensive in high-dimensional problems.
Second, we develop a projected zeroth-order ascent
method as a computationally scalable alternative,
without a global optimality guarantee.

We first present the distributed inner oracle and
its convergence analysis, followed by the two
outer optimization methods.
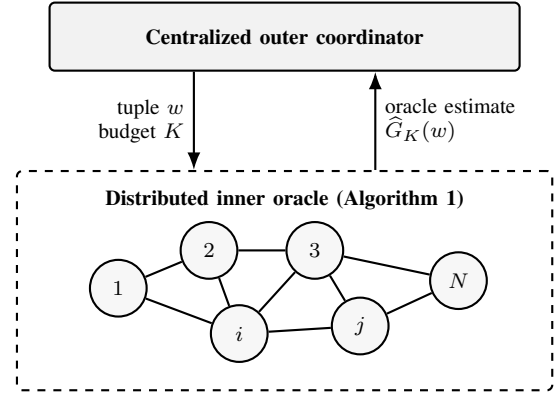
\begin{figure}[t]
\centering
\begin{tikzpicture}[
  font=\footnotesize,
  >=Latex,
  coordinator/.style={
    draw,
    rounded corners=2pt,
    thick,
    fill=black!5,
    minimum width=.72\columnwidth,
    minimum height=9mm,
    align=center
  },
  player/.style={
    circle,
    draw,
    thick,
    fill=black!3,
    minimum size=7.5mm,
    inner sep=0pt
  },
  oracle/.style={
  draw,
  rounded corners=2pt,
  thick,
  dashed,
  minimum width=.82\columnwidth,
  minimum height=29mm
  },
  flow/.style={->,thick}
]

\node[coordinator] (outer)
  {\textbf{Centralized outer coordinator}};

\node[oracle, below=13mm of outer] (inner) {};

\node[
  font=\footnotesize\bfseries,
  anchor=north,
  yshift=-1.5mm
] at (inner.north)
  {Distributed inner oracle (Algorithm~1)};

\node[player, at={($(inner.center)+(-22mm,-1mm)$)}] (p1) {$1$};
\node[player, at={($(inner.center)+(-10mm, 4mm)$)}] (p2) {$2$};
\node[player, at={($(inner.center)+(  4mm, 4mm)$)}] (p3) {$3$};
\node[player, at={($(inner.center)+( -6mm,-7mm)$)}] (pi) {$i$};
\node[player, at={($(inner.center)+( 10mm,-6mm)$)}] (pj) {$j$};
\node[player, at={($(inner.center)+( 23mm, 0mm)$)}] (pN) {$N$};

\draw[thick]
  (p1)--(p2)
  (p1)--(pi)
  (p2)--(p3)
  (p2)--(pi)
  (p3)--(pi)
  (p3)--(pj)
  (p3)--(pN)
  (pi)--(pj)
  (pj)--(pN);
  
\draw[flow]
  ([xshift=-12mm]outer.south) --
  node[left,align=right]{tuple $w$\\budget $K$}
  ([xshift=-12mm]inner.north);

\draw[flow]
  ([xshift=12mm]inner.north) --
  node[right,align=left]{oracle estimate\\$\widehat G_K(w)$}
  ([xshift=12mm]outer.south);

\end{tikzpicture}
\caption{Centralized outer search coupled with the distributed inner oracle.}
\label{fig:outer-inner}
\end{figure}
For concision, we give only the main proof arguments in this section;
complete proofs are provided in Appendix~\ref{app:section-iv-proofs}.

\subsection{Inner Distributed Optimization Oracle}
\label{sec:distributed-oracle}

Fix a candidate tuple $w=(y_1,\ldots,y_{d+1})\in\mathcal W$.
Using only its local utility information, each player
\(i\in\mathcal I\) defines the finite-tuple objective
\begin{equation}
\label{eq:local-tuple-objective}
h_i^w(\theta):=\max_{1\leq j\leq d+1}\ell_i(\theta,y_j).
\end{equation}
Then $F(\theta,w)=\max_{i\in \mathcal I} h_i^w(\theta)$ and
$G(w)=\min_{\theta\in\Theta}\max_{i\in \mathcal I} h_i^w(\theta)$.
Each $h_i^w$ is convex as a finite maximum of affine
functions. Player $i$ can compute its values and subgradients
using its own utility $u_i$, the feature map $\psi$,
and the tuple $w$.

Introducing an epigraph variable replaces the maximum
across players with one convex inequality per player. 
Choose common finite bounds $t_{\min}$ and $t_{\max}$ such that
\(
t_{\min}\leq h_i^w(\theta)\leq t_{\max}
\)
for all $i\in\mathcal I$, $\theta\in\Theta$, and $w\in\mathcal W$,
and let $\mathcal T:=[t_{\min},t_{\max}]$. Such bounds exist because
$h_i^w(\theta)$ is uniformly bounded by continuity and compactness;
the shared bounds used in implementation need not be tight.
We obtain
\begin{equation}
\label{eq:tuple-inner-epigraph}
\begin{aligned}
G(w)=\min_{\theta\in\Theta,\;t\in\mathcal T}\quad &t\\
\mathrm{s.t.}\quad &h_i^w(\theta)\leq t,
\qquad \forall i\in\mathcal I.
\end{aligned}
\end{equation}
Given $(\theta,t)$, player $i$ can evaluate its own constraint
$h_i^w(\theta)\leq t$.
The optimization requires finding a common pair satisfying
every player's constraint.
Suppose the players communicate over a connected undirected
graph $\mathcal G=(\mathcal I,\mathcal E_{\mathcal G})$.
Assign each player a local copy
$z_i:=(\theta_i^\top,t_i)^\top\in\mathcal Z:=\Theta\times\mathcal T.$
Enforcing neighbor consensus yields
\begin{equation}
\label{eq:distributed-tuple-inner}
\begin{aligned}
\min_{\{z_i\in\mathcal Z\}_{i\in\mathcal I}}\quad
&\frac1N\sum_{i=1}^N t_i\\
\mathrm{s.t.}\quad
&h_i^w(\theta_i)\leq t_i,\qquad \forall i\in\mathcal I,\\
&z_i=z_j,\qquad \forall(i,j)\in\mathcal E_{\mathcal G}.
\end{aligned}
\end{equation}
At consensus, its objective reduces to the common epigraph
value \(t\).

\par\smallskip
\begin{proposition}[Equivalence of the consensus formulation]
\label{prop:distributed-equivalence}
If $\mathcal G$ is connected, the feasible sets of
\eqref{eq:tuple-inner-epigraph} and
\eqref{eq:distributed-tuple-inner} correspond under replication
$z_i=(\theta,t)$ for every $i\in \mathcal I$, with identical objective values.
Therefore, their optimal values equal $G(w)$, and their
optimal solutions correspond to this replication.
\end{proposition}

\subsubsection{Applying existing distributed optimization oracle}
\label{subsec:relation-constrained}

Problem~\eqref{eq:distributed-tuple-inner} can be
solved using existing distributed optimization methods.
For fixed $w$, each constraint $h_i^w(\theta_i)\leq t_i$ is
equivalent to the $d+1$ affine inequalities
\begin{align}\label{eq: constraints}
\langle\Delta_\psi(y_j;i),\theta_i\rangle-t_i
\leq\Delta_{u_i}(y_j;i),
\quad j=1,\ldots,d+1.
\end{align}
These inequalities constitute an agent-specific affine conic
constraint with cone $\mathbb R_-^{d+1}$. Hence, the method
of~\cite{aybat2016primal} applies directly to each finite-tuple
subproblem. More recent methods such as
D-APDB~\cite{xu2026accelerated} accommodate more general private
functional constraints.

\begin{remark}
\label{rem:preserve-max}
Our proposed method retains the finite-maximum representation of
$h_i^w$. Applying a general affine-conic method to the expanded
formulation in \eqref{eq: constraints} would associate a $(d+1)$-dimensional constraint-dual
variable with each player, while our method uses one scalar
multiplier and obtains a subgradient from an active tuple element.
Hence, the local dual state is $\mathcal O(1)$ rather than
$\mathcal O(d)$, although evaluating the maximum may still require
examining all $d+1$ deviations. Our objective is not to obtain a
faster general-purpose convergence rate, but to construct a simpler
distributed oracle tailored to the finite-tuple subproblems.
\end{remark}

\subsubsection{Proposed distributed oracle}
Inspired from \cite{khuzani2017distributedregularizedprimaldualmethod,
Mateos2017saddlepoint,nedic2009subgradient}, we propose a primal dual update based oracle on regularized Lagrangian function for~\eqref{eq:distributed-tuple-inner}.
Let $\mu_i\geq0$ denote the multiplier associated with player
$i$'s constraint $h_i^w(\theta_i)\leq t_i$, and collect
$\mu:=(\mu_i)_{i\in\mathcal I}\in\mathbb R_+^N$.
For $\beta_k>0$, define
\begin{equation}
\label{eq:regularized-lagrangian}
\mathcal L_{\beta_k}(Z,\mu)
:=
\sum_{i\in\mathcal I}
\left[
\frac{t_i}{N}
+\mu_i\bigl(h_i^w(\theta_i)-t_i\bigr)
-\frac{\beta_k}{2}\mu_i^2
\right].
\end{equation}

Algorithm~1 implements distributed projected
descent--ascent on $\mathcal L_{\beta_k}$. 
At iteration $k$, let $Z[k]$ have rows $z_i[k]^\top$ and
let $\mu[k]:=(\mu_i[k])_{i\in\mathcal I}$.
Specifically, lines 4--6 compute \(
g_i[k]\in
\partial_{z_i}\mathcal L_{\beta_k}(Z[k],\mu[k]),\)
and line 7 computes the dual ascent direction \(
r_i[k]
=
\nabla_{\mu_i}\mathcal L_{\beta_k}(Z[k],\mu[k]).\)
Accordingly, line 8 performs projected primal descent with a
consensus correction, while line 9 performs projected dual ascent.
After the primal--dual iterations, lines~12--15 compute the
stepsize-weighted ergodic iterate of each player. Line~16 forms and
broadcasts the network-averaged potential parameter, line~17 evaluates
the corresponding oracle estimate, and line~18 returns the oracle
outputs.

\begin{algorithm}[t]
\caption{DRPDM Inner Oracle}
\label{alg:drpdm}
\begin{algorithmic}[1]
\Require $K\geq1$, $c,\beta_0>0$,
$\tau\in(0,1)$, $q\in(0,\tau/2)$,
$c\beta_0\leq1$, and
$\sigma\in\bigl(0,\lambda_{\max}^{-1}(L_{\mathcal G})\bigr]$;
set $\eta_k=c(k+1)^{-\tau}$ and
$\beta_k=\beta_0(k+1)^{-q}$.
\State Initialize $z_i[0]\in\mathcal Z$ and $\mu_i[0]=0$
\For{$k=0,\ldots,K-1$}
    \ForAll{$i\in\mathcal I$ \textbf{in parallel}}
        \State $j_i^\star[k]\in
        \arg\max_{1\leq j\leq d+1}\ell_i(\theta_i[k],y_j)$
        \State $s_i[k]\gets\Delta_\psi(y_{j_i^\star[k]};i)$
        \State $g_i[k]\gets
        \bigl(\mu_i[k]s_i[k],\,1/N-\mu_i[k]\bigr)$
        \State $r_i[k]\gets
        h_i^w(\theta_i[k])-t_i[k]-\beta_k\mu_i[k]$
        \State \(z_i[k+1]\gets
P_{\mathcal Z}\!\Bigl[z_i[k]-\eta_kg_i[k]\)
\Statex \hspace{\algorithmicindent}
\(\displaystyle
+\sigma\sum_{j\in\mathcal N_i}
[A_{\mathcal G}]_{ij}
\bigl(z_j[k]-z_i[k]\bigr)\Bigr]\)
        \State $\mu_i[k+1]\gets
        [\mu_i[k]+\eta_kr_i[k]]_+$
    \EndFor
\EndFor
\State $S_K\gets\sum_{k=0}^{K-1}\eta_k$
\ForAll{$i\in\mathcal I$ \textbf{in parallel}}
    \State $\bar z_i[K]\gets
    S_K^{-1}\sum_{k=0}^{K-1}\eta_kz_i[k]
    =:(\bar\theta_i[K],\bar t_i[K])$
\EndFor
\State $\bar\theta_w[K]\gets N^{-1}\sum_i\bar\theta_i[K]$;
broadcast $\bar\theta_w[K]$
\State Evaluate \(h_i^w(\bar\theta_w[K])\) locally; aggregate
\(\widehat G_K(w)\gets\max_i h_i^w(\bar\theta_w[K])\)
\State \Return $\bar\theta_w[K]$, $\widehat G_K(w)$,
and $\{\bar z_i[K]\}_{i\in\mathcal I}$
\end{algorithmic}
\end{algorithm}

\subsubsection{Finite-time analysis}
\label{subsec:finite-time-analysis}
The finite-budget output of Algorithm~\ref{alg:drpdm} is
generally not the exact tuple value $G(w)$. Since
$\bar\theta_w[K]\in\Theta$, the definition of $G(w)$ gives
\begin{equation}
\widehat G_K(w)
:=
F(\bar\theta_w[K],w)
\geq G(w).
\label{eq:finite-budget-upper-estimate}
\end{equation}
Thus, $\widehat G_K(w)$ is an upper estimate of $G(w)$ generated
by the finite-budget inner solver. The following result controls
this overestimation uniformly over $w\in\mathcal W$.

\par\smallskip
\begin{theorem}[Uniform inner-oracle accuracy]
\label{thm:drpdm-convergence}
Under the standing assumptions, run Algorithm~\ref{alg:drpdm}
with $\tau=3/4$ and $q=1/4$.  Then, for every \(w \in \mathcal W\)
\begin{equation}
\label{eq:uniform-oracle-rate}
0\leq \widehat G_K(w)-G(w)
=\widetilde{\mathcal O}\!\left(K^{-1/4}\right).
\end{equation}
\end{theorem}

\begin{proof}[Proof Sketch]
Define \(
\mathcal L(Z,\mu)
:=
\sum_{i\in\mathcal I}
\left[
\frac{t_i}{N}
+\mu_i\bigl(h_i^w(\theta_i)-t_i\bigr)
\right].\)
Let \(\bar Z[K]\) have rows \(\bar z_i[K]^\top\), and define
\(
\bar\mu[K]
:=
(S_K)^{-1}
\sum_{k=0}^{K-1}\eta_k\mu[k].
\)
Fix \(w\in\mathcal W\), and let \(
\theta_w^\star\in\arg\min_{\theta\in\Theta}F(\theta,w), \,
z_w^\star:=(\theta_w^\star,G(w)),
Z_c:=\mathbf 1(z_w^\star)^\top.\)
Since \(z_w^\star\) is feasible and
\(\bar\mu[K]\in\mathbb R_+^N\), \(
\mathcal L(Z_c,\bar\mu[K])\leq G(w).\)

To control the disagreement among the local primal iterates, we first
use the following bound on the dual iterates.
\par\smallskip
\begin{lemma}[Polynomial bound on dual iterates]
\label{lem:dual-iterate-bound}
Under the standing assumptions, there exists a constant
\(B_\mu>0\), independent of \(w\), \(i\), and \(k\), such that
\begin{equation}
    0\leq \mu_i[k]
    \leq \frac{B_\mu}{\beta_k}
    =\frac{B_\mu}{\beta_0}(k+1)^q,
    \qquad i\in\mathcal I,\quad k\geq0.
\end{equation}
\end{lemma}

Since the primal update direction \(g_i[k]\) depends linearly on
\(\mu_i[k]\), this bound implies
\(\|g_i[k]\|_2=\mathcal O((k+1)^q)\), which is the key input to the
following consensus estimate.

Define
\[
z_{\mathrm{av}}[k]
:=
\frac{1}{N}\sum_{i=1}^N z_i[k],
\quad
e_k
:=
\sqrt{
\frac{1}{N}\sum_{i=1}^N
\|z_i[k]-z_{\mathrm{av}}[k]\|_2^2
}.
\]
For the weighted ergodic iterates $\bar z_i[K]$ in Algorithm~\ref{alg:drpdm}, define
\[
\bar z_w[K]
:=
\frac{1}{N}\sum_{i\in\mathcal I}\bar z_i[K],
\qquad
D_K
:=
\frac{\sqrt{N}}{S_K}
\sum_{k=0}^{K-1}\eta_k e_k.
\]

\par\smallskip
\begin{lemma}[Consensus error]
\label{lem:consensus-error}
Under the standing assumptions, there exist constants
\(C_e,C_D>0\), independent of \(w\), \(K\), and the initial
primal points, such that
\begin{equation}
e_k\leq C_e(k+1)^{-1/2},
\qquad k\geq0,
\label{eq:pointwise-consensus}
\end{equation}
and for \(K\geq 1,\)
\begin{equation}
\max_{i\in\mathcal I}
\|\bar z_i[K]-\bar z_w[K]\|_2
\leq
D_K
\leq C_DK^{-1/4}.
\label{eq:ergodic-consensus}
\end{equation}
\end{lemma}

Let \(i_K\) attain
\(
F(\bar\theta_w[K],w)
=
h_{i_K}^w(\bar\theta_w[K]),
\)
and let \(\mathbf e_{i_K}\) be the corresponding coordinate unit
vector. Then, by Lemma~\ref{lem:consensus-error},
\[
F(\bar\theta_w[K],w)
\leq
\mathcal L(\bar Z[K],\mathbf e_{i_K})
+(G_\psi+1)D_K.
\]
It remains to control the Lagrangian gap in the preceding bound.
For this, we use the following ergodic primal--dual inequality.
\par\smallskip
\begin{lemma}[Ergodic primal--dual inequality]
\label{lem:ergodic-primal-dual}
Under the conditions of
Theorem~\ref{thm:drpdm-convergence}, let
\(
Z_c:=\mathbf 1z^\top
\)
be any consensus comparator with \(z\in\mathcal Z\). Then, for every
\(\mu\in\mathbb R_+^N\),
\begin{equation}
\mathcal L(\bar Z[K],\mu)
-
\mathcal L(Z_c,\bar\mu[K])
\leq
R_K+Q_K\|\mu\|_2^2,
\label{eq:ergodic-primal-dual}
\end{equation}
where, for some constant \(C_0>0\) independent of \(w\), \(K\),
and the initial primal points, \(
R_K
:=
\frac{C_0}{S_K}
\left[
1+
\sum_{k=0}^{K-1}\eta_k^2(k+1)^{1/2}
+
\sum_{k=0}^{K-1}\eta_k(k+1)^{1/4}e_k
\right],\)
and \(
Q_K
:=
\frac{1}{S_K}
\left[
\frac{1}{2}
+
\frac{1}{4}
\sum_{k=0}^{K-1}\eta_k\beta_k
\right].\)
\end{lemma}

Applying Lemma~\ref{lem:ergodic-primal-dual} with
\(\mu=\mathbf e_{i_K}\), and using
\(\|\mathbf e_{i_K}\|_2=1\), yields \(
0
\leq
\widehat G_K(w)-G(w)
\leq
R_K+Q_K+(G_\psi+1)D_K.\)

Finally, under \(\tau=3/4\) and \(q=1/4\),
\(S_K=\Theta(K^{1/4})\). The standard harmonic-sum bounds
applied to the definitions above then give \(
R_K=\widetilde{\mathcal O}(K^{-1/4}),
Q_K=\widetilde{\mathcal O}(K^{-1/4}).\)
Together with \(D_K=\mathcal O(K^{-1/4})\), this proves
\[
0\leq \widehat G_K(w)-G(w)
=\widetilde{\mathcal O}(K^{-1/4})
\]
uniformly over \(w\in\mathcal W\).
\end{proof}

\subsection{Outer Optimization}
\label{subsec:outer-optimization}

\subsubsection{Global Search with EXOTIC}
\label{subsec:exotic}

We first consider EXOTIC~\cite{EXOTIC} as a globally
convergent method for solving the outer problem
\eqref{eq:outer-optimal-value}.
EXOTIC adaptively explores a hierarchical partition of
$\mathcal W$. At a candidate tuple $w$, it invokes the
distributed inner oracle with an iteration budget $K$ and
obtains an approximate value $\widehat G_K(w)$.
It progressively allocates computation between exploring
new regions of $\mathcal W$ and refining the accuracy of
promising oracle evaluations.

Let $n$ denote the total computational budget, measured
as the cumulative number of synchronized inner iterations
used across all oracle queries, including the final
re-evaluations. Given $n$, EXOTIC determines a search
horizon $h_{\max}$ and adaptively allocates inner iteration
budgets across the nodes of the partition tree. After the
tree-search phase, selected candidates are re-evaluated
with a larger budget $K_{\mathrm{final}}$, and the tuple
with the largest refined estimate is returned.
We defer the complete EXOTIC procedure and its budget
allocation rule to the Appendix~\ref{app:exotic-algorithm}.

To invoke the convergence guarantee of EXOTIC from \cite{EXOTIC}, we verify
two properties of the outer problem. First, the objective
$G$ must vary regularly within cells of the hierarchical
partition. Second, the number of cells that remain
near-optimal at increasingly fine resolutions must be
controlled. We establish these properties below.

Equip $\mathcal W$ with the metric \(
d_{\mathcal W}(w,\widetilde w)
:=
\max_{1\leq j\leq d+1}
\|y_j-\widetilde y_j\|_2,\)
and let $L_u$ and $L_\psi$ denote common Lipschitz
constants of the utilities and the feature map,
respectively.

\begin{lemma}[Regularity of the outer objective]
\label{lem:outer-lipschitz}
Let \(
C_\theta:=\max_{\theta\in\Theta}\|\theta\|_2.\)
Then, for all $w,\widetilde w\in\mathcal W$,
\begin{equation}
\label{eq:outer-lipschitz}
|G(w)-G(\widetilde w)|
\leq
L_{\mathrm{game}}
d_{\mathcal W}(w,\widetilde w),
\end{equation}
where \(
L_{\mathrm{game}}
:=
2(C_\theta L_\psi+L_u).\)
\end{lemma}

EXOTIC uses a $B$-ary hierarchical partition
$\{\mathcal W_{h,r}\}_{r=1}^{B^h}$ of $\mathcal W$,
where $h$ denotes the depth. Assume that the cell
diameters decrease geometrically:
\begin{equation}
\operatorname{diam}_{d_{\mathcal W}}
(\mathcal W_{h,r})
\leq
\alpha_{\mathcal W}\rho^h,
\qquad
\alpha_{\mathcal W}>0,\quad 0<\rho<1.
\label{eq:partition-diameter}
\end{equation}
By Lemma~\ref{lem:outer-lipschitz}, if
$\mathcal W_{h,r}$ contains an outer maximizer $w^\star$,
then
\[
G(w^\star)-G(w)
\leq
\nu_{\mathrm{game}}\rho^h,
\qquad
w\in\mathcal W_{h,r},
\]
where \(
\nu_{\mathrm{game}}
:=
\alpha_{\mathcal W}L_{\mathrm{game}}.\)
Thus, the variation of $G$ within an optimal cell
decreases geometrically with the partition depth.

It remains to characterize how many cells can remain
near-optimal at a given depth. Define \(
\mathcal N_h(\varepsilon)
:=
\left\{
r:
\max_{w\in\mathcal W_{h,r}}G(w)
\geq p^\star-\varepsilon
\right\}.\)
Following~\cite[Definition~3.2]{EXOTIC}, let
$d_G\geq0$ be an admissible near-optimality exponent if,
for some $C_G>1$,
\begin{equation}
\left|
\mathcal N_h(2\nu_{\mathrm{game}}\rho^h)
\right|
\leq
C_G\rho^{-d_Gh},
\qquad h\geq0.
\label{eq:near-optimality-dimension}
\end{equation}
The exponent $d_G$ captures the effective difficulty of
the outer search: larger $d_G$ allows more cells to remain
competitive as the partition is refined.
We can now combine these outer regularity properties with
the uniform accuracy of the distributed inner oracle.

\begin{theorem}[Global outer-value guarantee]
\label{thm:exotic-global}
Suppose that the assumptions of
Theorem~\ref{thm:drpdm-convergence} hold and that the
hierarchical partition satisfies
\eqref{eq:partition-diameter} and
\eqref{eq:near-optimality-dimension}.
Run EXOTIC with Algorithm~\ref{alg:drpdm} as its inner
oracle and total inner-iteration budget $n$.
Then, for the returned tuple $\widehat w_n$,
\[
0
\leq
p^\star-G(\widehat w_n)
=
\widetilde{\mathcal O}
\left(
n^{-1/(d_G+4)}
\right),
\]
with all problem and algorithm constants fixed.
\end{theorem}

\begin{proof}[Proof sketch]
Theorem~\ref{thm:drpdm-convergence} gives the uniform
one-sided oracle bound \(
0
\leq
\widehat G_K(w)-G(w)
=
\widetilde{\mathcal O}(K^{-1/4})\), for every \(
w\in\mathcal W.\)
Lemma~\ref{lem:outer-lipschitz} together with
\eqref{eq:partition-diameter} verifies the geometric
variation condition required by EXOTIC, while
\eqref{eq:near-optimality-dimension} controls the number
of near-optimal cells at each depth. Substituting the
inner-oracle exponent $1/4$ into
\cite[Theorem~3.3 and Corollary~3.5]{EXOTIC}
gives the stated rate, with logarithmic factors absorbed
into $\widetilde{\mathcal O}$.
\end{proof}

Theorem~\ref{thm:exotic-global} guarantees convergence
of the \emph{outer value} $G(\widehat w_n)$ to
$p^\star$. The associated potential parameter is obtained
from the final inner solve, \(
\widehat\theta_n
:=
\bar\theta_{\widehat w_n}[K_{\mathrm{final}}].
\)
The theorem does not, by itself, provide a finite-budget
bound on the full-domain error
$\mathcal E(\widehat\theta_n)$.
Exact optimizer recovery is characterized by
Proposition~\ref{prop:reformulation-recovery}; obtaining
a quantitative finite-budget recovery guarantee would
require additional stability or growth conditions.
\subsubsection{Projected Zeroth-Order Ascent}
\label{subsec:zeroth-order}
Although EXOTIC provides a global outer-search guarantee, its
hierarchical partition can become computationally expensive when
$\mathcal W$ is high-dimensional. We therefore also consider a
projected zeroth-order ascent method that uses only evaluations of
the distributed oracle and avoids maintaining a partition tree.

Assume that $\mathcal W\subseteq\mathbb R^{D_{\mathcal W}}$ is convex
with nonempty interior, and define the $\delta$-interior \(
\mathcal W_\delta
:=
\left\{
w\in\mathcal W:
w+\delta v\in\mathcal W
\text{ for all } \|v\|_2\leq1
\right\}.\)
At iteration $t$, we sample \(
v_t\sim\operatorname{Unif}(\mathbb S^{D_{\mathcal W}-1})\)
and query the distributed oracle at
$w_t$ and
\(
w_t^+:=w_t+\delta_t v_t.
\)
Let $\widehat G_t^0$ and $\widehat G_t^+$ denote the corresponding
oracle estimates. We construct the zeroth-order ascent direction \(
\widehat g_t
:=
\frac{D_{\mathcal W}}{\delta_t}
\bigl(
\widehat G_t^+-\widehat G_t^0
\bigr)v_t,\)
and update \(
w_{t+1}
:=
P_{\mathcal W_{\delta_{t+1}}}
\left(
w_t+\alpha_t\widehat g_t
\right),\)
where $\delta_t$ is the perturbation radius and $\alpha_t$ is the
outer stepsize.

Unlike EXOTIC, this procedure does not perform a global search over
$\mathcal W$ and requires only two oracle evaluations per outer
iteration. It therefore provides a more scalable alternative for
high-dimensional problems. In this work, however, we use the method
as a heuristic and do not establish global-optimality or stationarity
guarantees for the resulting inexact projected updates.
The complete procedure, including the decay schedules for
$\delta_t$ and $\alpha_t$ and the final re-evaluation step, is given in Appendix~\ref{app:zo-algorithm}.

\section{Numerical Experiments}
\label{sec:numerical-experiments}
We evaluate the framework in three complementary settings.
A two-player benchmark with known optimum compares EXOTIC and
zeroth-order (ZO) search and assesses parameter recovery; a
three-player linear--quadratic game tests improvement over an
analytical $\alpha$-potential construction; and a 50-player game
examines the scalability of ZO search in a high-dimensional
potential class.

\subsection{Experimental Protocol}
\label{subsec:experimental-protocol}

All experiments use Algorithm~\ref{alg:drpdm} as the distributed
inner solver with
\(c=1\), \(\tau=\tfrac12\), \(\beta_0=0.125\), and \(q=\tfrac16\),
unless stated otherwise; the consensus stepsize $\sigma$ is chosen
separately for each communication graph. Numerically, we use the
uniform ergodic average \(
\bar z_i^{\mathrm{num}}[K]
:=
K^{-1}\sum_{k=0}^{K-1}z_i[k],\)
instead of the stepsize-weighted average used in the analysis.

For EXOTIC, we use a deterministic binary midpoint partition with
\(h_{\max}\in\{500,1000,2000\}\) and
\(s\in\{10^3,10^4\}\), where a nominal inner budget $K$
corresponds to $sK$ DRPDM iterations. For ZO search, we set
\(\delta_t=\delta_0(t+1)^{-1/4}\) and
\(\alpha_t=\alpha_0(t+1)^{-4/5}\), and cold-start each inner solve.
Experiment-specific values of $\alpha_0$, $\delta_0$, and the
iteration budgets are given below.
\subsection{Two-Player Benchmark with Known Optimum}
\label{subsec:two-player-experiment}

Our first experiment considers a controlled benchmark for which the
optimal full-domain approximation and corresponding potential parameter
can be computed exactly. This allows us to evaluate both objective and
parameter recovery and to compare the two outer-search methods.

Consider a two-player, two-action game, where \(x_i\in[0,1]\) denotes
the probability that player \(i\) selects its first action. The
mixed-strategy utilities are
\(
u_1(x_1,x_2)
=12-4x_1-10x_2+10x_1x_2
\)
and
\(
u_2(x_1,x_2)
=10-2x_1-8x_2+2x_1x_2.
\)
We use
\(
\psi(x_1,x_2)=(x_1,x_2,x_1x_2)^\top
\)
and
\(
\phi_\theta(x_1,x_2)
=
\langle\theta,\psi(x_1,x_2)\rangle.
\)
Because these utilities are multilinear extensions of a finite
\(2\times2\) game, the optimal approximation can be computed exactly
using the finite-game linear program of~\cite{candogan2011near}.
For the feature class above, this yields
\(
\theta^\star=(-2,-10,6)^\top
\)
and
\(
p^\star=2.
\)
Consequently, we can evaluate both the full-domain optimality gap
\(\mathcal E(\widehat\theta)-p^\star\) and the parameter-recovery error
\(\|\widehat\theta-\theta^\star\|_2\).

The players communicate over a single undirected edge. We set
\(\sigma=0.45\) and \(\Theta=[-25,25]^3\). Since
\(x_i,x_i'\in[0,1]\), the resulting outer search space is
\(\mathcal W=\mathcal Y^{d+1}=[0,1]^{16}\).
For ZO search, we use \(T=2000\),
\(K_{\mathrm{inner}}=10^4\), and
\(K_{\mathrm{reeval}}=10^5\), and test
\(
\alpha_0\in\{0.01,0.03,0.1\}
\)
and
\(
\delta_0\in\{0.0125,0.025,0.05\}.
\)
For EXOTIC, we use the parameter grid specified in
Section~\ref{subsec:experimental-protocol}.

\begin{figure}[h]
    \centering
    \includegraphics[width=0.49\columnwidth]
    {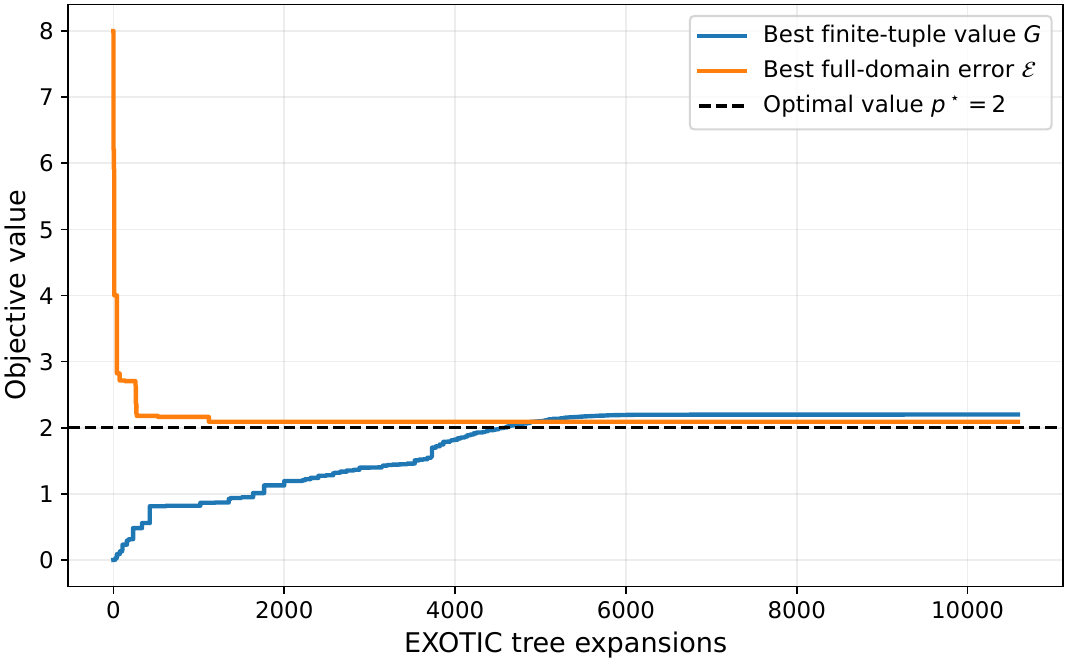}
    \hfill
    \includegraphics[width=0.48\columnwidth]
    {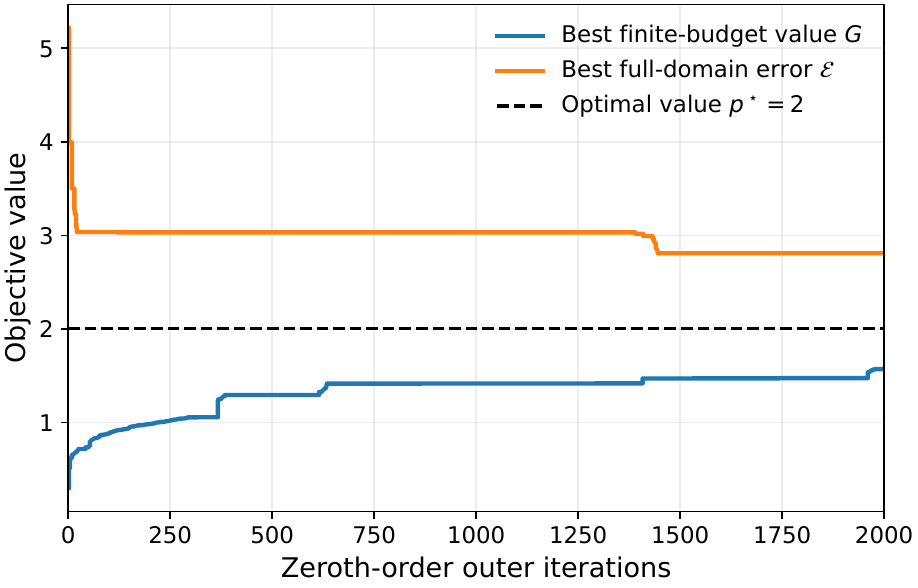}
    \caption{Representative best-so-far finite-tuple values and exact
    full-domain errors for the two-player game: DRPDM-based EXOTIC
    (\(h_{\max}=2000\), \(s=10^4\); left) and ZO
    (\(\alpha_0=0.1\), \(\delta_0=0.025\); right).
        The dashed lines mark \(p^\star=2\).}
    \label{fig:n2-best-so-far-comparison}
\end{figure}

Table~\ref{tab:n2-results} compares the two methods.
For EXOTIC, increasing \(h_{\max}\) consistently reduces the recovered
full-domain error. On this low-dimensional benchmark, EXOTIC achieves
more accurate objective and parameter recovery, at the cost of greater
computation at its highest-accuracy setting.
Figure~\ref{fig:n2-best-so-far-comparison} shows representative
trajectories for EXOTIC with \((h_{\max},s)=(2000,10^4)\) and ZO with
\((\alpha_0,\delta_0)=(0.1,0.025)\); the latter is shown for
illustration and is not the best-performing ZO configuration in
Table~\ref{tab:n2-results}.

\begin{table}[t]
\centering
\caption{Performance on the two-player continuous game. For ZO, the
configuration denotes $(\alpha_0,\delta_0)$; for EXOTIC, it denotes
$(h_{\max},s)$.}
\label{tab:n2-results}
\resizebox{\columnwidth}{!}{%
\begin{tabular}{llccccc}
\toprule
Method
& Configuration
& $\widehat G$
& $\mathcal{E}(\widehat\theta)$
& $\mathcal{E}(\widehat\theta)-p^\star$
& $\|\widehat\theta-\theta^\star\|_2$
& Time (s) \\
\midrule

ZO
& $(10^{-2},0.0125)$
& 0.3063
& \textbf{2.3436}
& \textbf{0.3436}
& \textbf{0.4650}
& 12.68 \\

& $(3\times10^{-2},0.0125)$
& 0.4363
& 2.4676
& 0.4676
& 1.0612
& 12.43 \\

& $(10^{-1},0.025)$
& \textbf{0.6550}
& 2.4251
& 0.4251
& 0.8502
& 12.30 \\

\midrule

EXOTIC
& $(500,10^4)$
& 1.6265
& 2.4512
& 0.4512
& 0.5217
& 6.98 \\

& $(1000,10^4)$
& 1.9006
& 2.1555
& 0.1555
& 0.1166
& 16.39 \\

& $(2000,10^4)$
& \textbf{1.9214}
& \textbf{2.0655}
& \textbf{0.0655}
& \textbf{0.1088}
& 38.37 \\

\bottomrule
\end{tabular}%
}
\end{table}

\subsection{Linear--Quadratic Network Games}
\label{subsec:lq-experiments}

We next consider network games for which an analytical
$\alpha$-potential construction is available from
\cite{rokade2025asymmetric}. Rather than recovering a known optimum,
these experiments test whether our optimization-based approach can
improve upon this analytical construction. We consider a small
instance amenable to EXOTIC and a larger instance for which we use
the ZO outer method. Consider the $N$-player linear--quadratic network game
\begin{equation}
\label{eq:lq-utility}
u_i(a)
=
-\frac12 a_i^2+\zeta_i a_i
+\gamma a_i\sum_{j=1}^N V_{ij}a_j,
\qquad
a_i\in[-1,1],
\end{equation}
where \(V_{ii}=0\) and the off-diagonal entries are sampled
independently as
\(V_{ij}\sim\operatorname{Unif}[1-\varepsilon_N,1+\varepsilon_N]\).
The players communicate over an undirected cycle.
We optimize over the quadratic potential class
\begin{equation}
\label{eq:lq-potential-class}
\phi_\theta(a)
=
\sum_{i=1}^N p_i a_i^2
+
\sum_{i=1}^N q_i a_i
+
\sum_{1\leq i<j\leq N} r_{ij}a_i a_j,
\end{equation}
where
\(
\theta
=
\bigl(
\{p_i\}_{i=1}^N,\,
\{q_i\}_{i=1}^N,\,
\{r_{ij}\}_{1\leq i<j\leq N}
\bigr)
\in\Theta\subseteq\mathbb R^d
\)
and
\(
d=2N+\frac{N(N-1)}{2}.
\)
As the analytical baseline, we use
\begin{equation}
\label{eq:lq-initialization}
p_i^{(0)}=-\frac12,
\qquad
q_i^{(0)}=\zeta_i,
\qquad
r_{ij}^{(0)}
=
\frac{\gamma}{2}
\bigl(V_{ij}+V_{ji}\bigr),
\quad i<j.
\end{equation}
By
\cite{rokade2025asymmetric}, its full-domain approximation error is
\(
\mathcal E(\theta_0)
=
|\gamma|\,\|V-V^\top\|_\infty.
\)
For this class of games, \(\mathcal E(\theta)\) can be evaluated
exactly over the continuous action domain using a closed-form
expression provided in Appendix~\ref{app:lq-full-domain-error}.
We measure improvement by
\begin{equation}
\label{eq:lq-improvement}
\operatorname{Imp}(\widehat\theta)
:=
\mathcal E(\theta_0)-\mathcal E(\widehat\theta),
\qquad
\operatorname{Imp}_{\%}(\widehat\theta)
:=
100\frac{\operatorname{Imp}(\widehat\theta)}
{\mathcal E(\theta_0)}.
\end{equation}
Positive values therefore indicate improvement over the analytical
construction in the original full-domain objective.
\subsubsection{Three-Player Game: Improving the Analytical Construction}
\label{subsubsec:lq-n3-exotic}
We first consider a low-dimensional instance for which EXOTIC can be
used for global search over the finite-tuple space. The goal is to test
whether optimizing the finite-tuple objective can improve upon the
analytical baseline in~\eqref{eq:lq-initialization}.
We consider a three-player instance with
\(\gamma=\zeta_i=1\), \(\varepsilon_N=0.5\), and communication over
the three-node cycle. We consider the following parameterization of potential
\begin{equation*}
\phi_\theta(a)
=
-\frac12\sum_{i=1}^3a_i^2
+\sum_{i=1}^3a_i
+\sum_{1\le i<j\le3}\theta_{ij}a_i a_j,
\quad
\theta\in[-2,2]^3.
\end{equation*}
Thus, \(d_\theta=3\), \(d_y=6\), and
\(\mathcal W=\mathcal Y^4\subseteq\mathbb R^{24}\).
The analytical baseline is
\(
\theta_{ij}^{(0)}=(V_{ij}+V_{ji})/2
\),
for which
\(
\mathcal E(\theta_0)=\|V-V^\top\|_\infty=0.4480.
\)
We use the EXOTIC parameter grid from
Section~\ref{subsec:experimental-protocol} and set \(\sigma=0.3\).

Across the tested configurations, the smallest observed exact
full-domain error is \(0.3693\), corresponding to a \(17.58\%\)
improvement over the analytical baseline.
Figure~\ref{fig:n3-exotic-best-so-far} shows that the best finite-tuple
value increases during the search, while the diagnostic full-domain
error initially decreases and then stabilizes. Thus, in this instance,
finite-tuple optimization identifies a potential that improves upon
the analytical construction over the full action domain.
\begin{figure}[t]
    \centering
    \includegraphics[width=0.5\columnwidth]
    {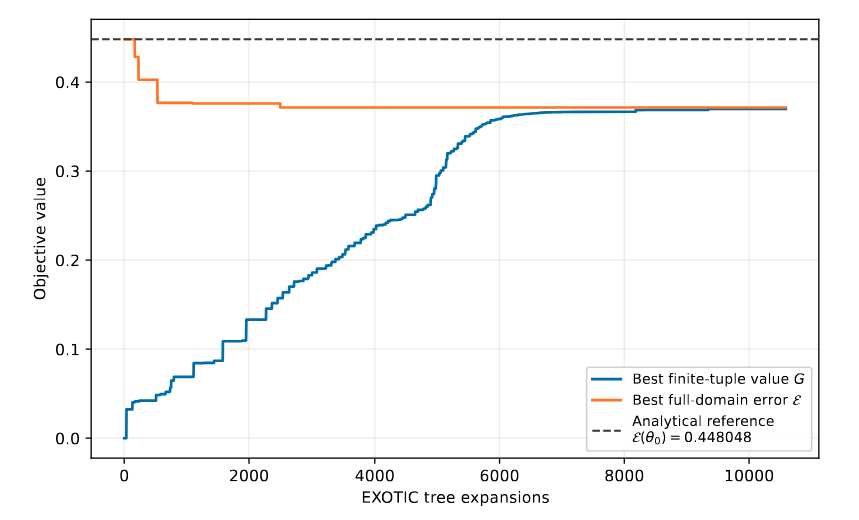}
    \caption{Representative EXOTIC search
    (\(h_{\max}=2000\), \(s=10^4\)). The dashed line marks
    \(\mathcal E(\theta_0)=\|V-V^\top\|_\infty\).}
    \label{fig:n3-exotic-best-so-far}
\end{figure}





\subsubsection{Fifty-Player Game: Scalable Zeroth-Order Search}
\label{subsubsec:lq-n50-zo}

We finally consider a larger instance to examine the scalability of
the ZO outer method. The full quadratic potential class has
\(d_\theta=1325\), making hierarchical global search over
\(\mathcal W=\mathcal Y^{d_\theta+1}\) computationally impractical.
We therefore combine the ZO outer search with the distributed DRPDM
oracle.

We consider a \(50\)-player instance with
\(\gamma=\zeta_i=1\) and \(\varepsilon_N=0.8\).
The interaction matrix \(V\) is generated with seed \(42\) and held
fixed throughout the experiment. The players communicate over a
\(50\)-node cycle, with \(\sigma=0.225\).
We use the full quadratic potential class in
\eqref{eq:lq-potential-class}, with
\(d_\theta=1325\) and \(\Theta=[-2,2]^{1325}\).
The analytical baseline in~\eqref{eq:lq-initialization} has exact
full-domain error
\(
\mathcal E(\theta_0)
=
\|V-V^\top\|_\infty
=
33.5072.
\)

We run Algorithm~\ref{alg:zeroth-order} for \(T=200\) iterations and
test
\(
\alpha_0\in\{10^{-2},10^{-3},10^{-4}\}
\)
and
\(
\delta_0\in\{0.00625,0.0125,0.025\}.
\)
Each base and perturbed query uses
\(K_{\mathrm{inner}}=10{,}000\) DRPDM iterations, with the two
queries evaluated in parallel from the same cold start \(\theta_0\).
All nine configurations share the same interaction matrix, initial
outer iterate, and zeroth-order directions.

Because the outer method optimizes the finite-tuple objective rather
than the full-domain error, we use the closed-form expression for
\(\mathcal E\) available in this example only for diagnostic
post-selection. Candidates whose initial full-domain error improves
upon \(\mathcal E(\theta_0)\) are independently re-solved using
\(K_{\mathrm{reeval}}=50{,}000\) DRPDM iterations, and we report the
candidate with the smallest resulting exact error. If no candidate
passes this screening step, we instead re-evaluate the candidate with
the smallest initial error.

Figure~\ref{fig:n50-zo-alpha-delta} shows the best-so-far trajectories,
and Table~\ref{tab:n50-zo-results} summarizes the post-selected
outputs. Four of the nine configurations improve upon the analytical
baseline. The best result is obtained with
\((\alpha_0,\delta_0)=(10^{-2},0.025)\), for which the iterate selected
at \(t=163\) achieves
\(
\mathcal E(\widehat\theta_{\mathrm{out}})=31.7768,
\)
corresponding to a \(5.16\%\) improvement. In contrast, all
configurations with \(\alpha_0=10^{-4}\) perform worse than the
baseline, suggesting that this outer stepsize is too small for the
available iteration budget.

The experiment also illustrates the distinction between outer-value recovery and parameter recovery. For an arbitrary nonoptimal tuple \(w\), increasing \(G(w)\) need not monotonically reduce \(\mathcal{E}(\theta_w)\), since \(\theta_w\) is optimized only against the deviations contained in \(w\). Theorem \ref{thm:exotic-global} therefore controls the outer value, while recovery of the corresponding full-domain optimal potential is governed asymptotically by the structural conditions of Proposition \ref{prop:reformulation-recovery}.
\begin{table}[h]
\centering
\caption{Post-selected outputs for the 50-player game with $c=1$,
after independent DRPDM re-evaluation with
$K_{\mathrm{reeval}}=50{,}000$.}
\label{tab:n50-zo-results}
\begin{tabular}{ccccc}
\toprule
$\alpha_0$ & $\delta_0$ & Selected $t$ &
$\mathcal{E}(\theta_{\mathrm{out}})$ & Improvement \\
\midrule
$10^{-2}$ & $0.0125$  & 140 & 32.730435 &
$\phantom{-}0.776793\;(2.32\%)$ \\
$\mathbf{10^{-2}}$ & $\mathbf{0.025}$ & \textbf{163} &
\textbf{31.776847} & $\mathbf{1.730381\;(5.16\%)}$ \\
\addlinespace
$10^{-3}$ & $0.0125$  & 67  & 33.584310 &
$-0.077082\;(-0.23\%)$ \\
$10^{-3}$ & $0.025$   & 160 & 32.786437 &
$\phantom{-}0.720791\;(2.15\%)$ \\
\bottomrule
\end{tabular}
\end{table}

\begin{figure}[h]
    \centering
    \includegraphics[
        width=0.48\textwidth
    ]{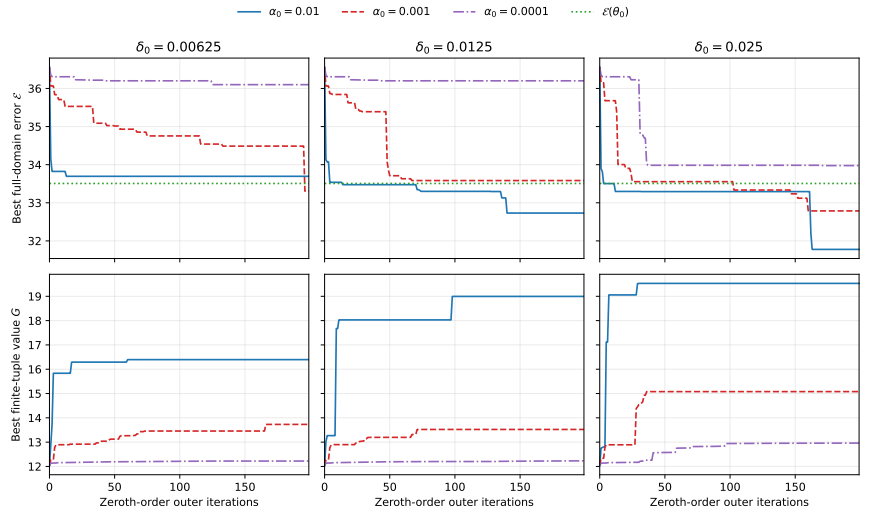}
    \caption{Best-so-far exact errors (top) and finite-tuple values
(bottom) for the \(50\)-player ZO experiment. Columns correspond to
\(\delta_0\in\{0.00625,0.0125,0.025\}\), and curves to
\(\alpha_0\in\{10^{-2},10^{-3},10^{-4}\}\). The dashed lines mark
\(\mathcal E(\theta_0)=33.5072\).}
    \label{fig:n50-zo-alpha-delta}
\end{figure}

\section{Conclusion}

We studied distributed computation of near-potential approximations
for continuous-action general-sum games. An exact finite-tuple
reformulation separates the problem into a global search over deviation
tuples and distributed convex inner problem. For the inner problem,
we developed a regularized primal--dual oracle with a uniform
one-sided accuracy guarantee. Combined with EXOTIC, this yields an
end-to-end global-search guarantee, while projected zeroth-order ascent
provides a more scalable alternative. We also established sufficient
conditions for recovering an optimizer of the original full-domain
problem. Numerical experiments illustrate the tradeoff between global
search accuracy and computational scalability.

\bibliographystyle{ieeetr}
\bibliography{sample}
\newpage 
\appendices

\section{Proof of Proposition~\ref{prop:reformulation-recovery}}
\label{app:reformulation-recovery-proof}

The value identity follows from the reformulation argument of~\cite[Proposition~2.1]{EXOTIC}. Indeed, \(\ell_i(\theta,y)\) is continuous in \((\theta,y)\) and affine, hence convex, in \(\theta\), while \(\Theta\) and \(\mathcal Y\) are compact. Compactness of \(\Theta\) replaces the coercivity assumption used in~\cite{EXOTIC} to guarantee boundedness and attainment. Therefore, \[ p^\star=\max_{w\in\mathcal W}G(w). \]

Let \(w^\star\in \operatorname*{arg\,max}_{w\in\mathcal W}G(w)\) and \(\theta_{w^\star}\in \operatorname*{arg\,min}_{\theta\in\Theta}F(\theta,w^\star)\). Then \[ F(\theta_{w^\star},w^\star) = G(w^\star) = p^\star. \]

For optimizer recovery, we verify the two alternative conditions
of~\cite[Proposition~2.4]{EXOTIC}.
Under condition~\emph{(i)}, \[ \mathcal E(\theta_{w^\star}) = F(\theta_{w^\star},w^\star) = p^\star, \] so \(\theta_{w^\star}\) solves the full-domain problem.

It remains to consider condition~\emph{(ii)}. For every \((i,j)\in S^\star\), define the active loss gradient \[ g_{ij}^\star:=\Delta_\psi(y_j^\star;i). \] 
Fix any
\(\theta\in\Theta\) with
\(\theta\ne\theta_{w^\star}^\star\), and let
\(
v:=\theta-\theta_{w^\star}^\star\ne0.
\)
For every \((i,j)\in\mathcal S^\star\), 
the affinity gives 
\[ 
\ell_i(\theta,y_j^\star) = \ell_i(\theta_{w^\star},y_j^\star) + \langle g_{ij}^\star,\theta-\theta_{w^\star}\rangle. 
\]  
Since every \((i,j)\in S^\star\) is active at \(\theta_{w^\star}\), 
\[ 
\begin{aligned} F(\theta,w^\star) &\geq \max_{(i,j)\in S^\star} \ell_i(\theta,y_j^\star)\\ 
&= F(\theta_{w^\star},w^\star) + \max_{(i,j)\in S^\star} \langle g_{ij}^\star,v\rangle. 
\end{aligned} 
\] 
We claim that 
\[ 
\max_{(i,j)\in S^\star} \langle g_{ij}^\star,v\rangle>0. 
\] 
Suppose instead that \(\langle g_{ij}^\star,v\rangle\leq 0\) 
for every \((i,j)\in S^\star\). 
Then the convex hull of the active gradients is contained in the closed halfspace 
\[ 
\{g\in\mathbb R^d:\langle g,v\rangle\leq 0\}. 
\] 
This contradicts condition~\emph{(ii)}, 
because 
\[ 
0\in \operatorname{int}\operatorname{conv} \{g_{ij}^\star:(i,j)\in S^\star\} 
\] 
implies that this convex hull contains a neighborhood of the origin, including a point in the direction \(v\) with strictly positive inner product with \(v\). Hence, 
\[ 
F(\theta,w^\star) > F(\theta_{w^\star},w^\star) \qquad \text{for every } \theta\in\Theta\setminus\{\theta_{w^\star}\}. 
\] 
Thus, \(\theta_{w^\star}\) is the unique minimizer of \(F(\cdot,w^\star)\). 

Finally, let
\(\theta^\star\in
\operatorname*{arg\,min}_{\theta\in\Theta}\mathcal E(\theta)\).
Since \(w^\star\) maximizes \(G\),
\[
p^\star
=
G(w^\star)
=
\min_{\theta\in\Theta}F(\theta,w^\star)
\leq
F(\theta^\star,w^\star).
\]
On the other hand, the tuple elements
\(y_1^\star,\ldots,y_{d+1}^\star\) belong to \(\mathcal Y\), so
\[
F(\theta^\star,w^\star)
=
\max_{i,j}\ell_i(\theta^\star,y_j^\star)
\leq
\max_{i,\;y\in\mathcal Y}\ell_i(\theta^\star,y)
=
\mathcal E(\theta^\star)
=
p^\star.
\]
Hence \(F(\theta^\star,w^\star)=G(w^\star)\), and therefore
\(\theta^\star\) also minimizes \(F(\cdot,w^\star)\).
Its uniqueness gives 
\(\theta^\star=\theta_{w^\star}\), and hence 
\[ 
\theta_{w^\star} \in \operatorname*{arg\,min}_{\theta\in\Theta}\mathcal E(\theta). 
\] 
This proves optimizer recovery under condition~\emph{(ii)}.

\section{Auxiliary Results}
\label{app:auxiliary-results}

\begin{lemma}
\label{lem:geometric-polynomial}
Let $\rho\in(0,1)$ and $\xi\in(0,1)$. Then, for every $k\ge1$,
\begin{equation}
\label{eq:geometric-polynomial}
\rho^k
\le
\left(
\frac{\xi}{e|\log\rho|}
\right)^\xi k^{-\xi},
\end{equation}
and
\begin{equation}
\label{eq:geometric-convolution}
\sum_{s=0}^{k-1}
\rho^{k-1-s}(s+1)^{-\xi}
\le
C_{\rho,\xi}k^{-\xi},
\end{equation}
where
\[
C_{\rho,\xi}
:=
\frac{2^\xi}{1-\rho}
+
\frac{2}{e(1-\xi)|\log\rho|}.
\]
\end{lemma}

\begin{proof}
Let $b:=|\log\rho|>0$ and define
$f(x):=x^\xi e^{-bx}$ for $x>0$. Since
\[
f'(x)=x^{\xi-1}e^{-bx}(\xi-bx),
\]
$f$ is increasing on $(0,\xi/b)$ and decreasing on
$(\xi/b,\infty)$. Hence,
\[
\sup_{x>0}x^\xi\rho^x
=
f\left(\frac{\xi}{b}\right)
=
\left(
\frac{\xi}{e|\log\rho|}
\right)^\xi.
\]
Taking $x=k$ proves \eqref{eq:geometric-polynomial}.

To prove \eqref{eq:geometric-convolution}, first rewrite the
convolution as
\[
\sum_{s=0}^{k-1}\rho^{k-1-s}(s+1)^{-\xi}
=
\sum_{t=1}^{k}\rho^{k-t}t^{-\xi}.
\]
We split the index set into
\[
\mathcal I_1:=\{t\in[k]:t<k/2\},
\qquad
\mathcal I_2:=\{t\in[k]:t\ge k/2\}.
\]

For $t\in\mathcal I_2$, we have
$t^{-\xi}\le 2^\xi k^{-\xi}$. Therefore,
\begin{align*}
\sum_{t\in\mathcal I_2}\rho^{k-t}t^{-\xi}
&\le
2^\xi k^{-\xi}
\sum_{t\in\mathcal I_2}\rho^{k-t}\\
&\le
\frac{2^\xi}{1-\rho}k^{-\xi}.
\end{align*}

For $t\in\mathcal I_1$, we have $k-t>k/2$ and hence
$\rho^{k-t}\le\rho^{k/2}$. Since $\xi\in(0,1)$,
\[
\sum_{t\in\mathcal I_1}t^{-\xi}
\le
\sum_{t=1}^{k}t^{-\xi}
\le
\int_0^k x^{-\xi}\,dx
=
\frac{k^{1-\xi}}{1-\xi}.
\]
It follows that
\begin{align*}
\sum_{t\in\mathcal I_1}\rho^{k-t}t^{-\xi}
&\le
\frac{\rho^{k/2}k^{1-\xi}}{1-\xi}\\
&=
\frac{k\rho^{k/2}}{1-\xi}k^{-\xi}.
\end{align*}
To bound the remaining factor, define
$g(x):=x\rho^{x/2}$ for $x>0$. Since
\[
g'(x)
=
\rho^{x/2}
\left(
1-\frac{|\log\rho|}{2}x
\right),
\]
the maximum of $g$ is attained at
$x=2/|\log\rho|$. Consequently,
\[
k\rho^{k/2}
\le
\sup_{x>0}x\rho^{x/2}
=
\frac{2}{e|\log\rho|}.
\]
Thus,
\[
\sum_{t\in\mathcal I_1}\rho^{k-t}t^{-\xi}
\le
\frac{2}{e(1-\xi)|\log\rho|}k^{-\xi}.
\]

Combining the bounds over $\mathcal I_1$ and $\mathcal I_2$ gives
\begin{align*}
\sum_{s=0}^{k-1}\rho^{k-1-s}(s+1)^{-\xi}
&\le
\left[
\frac{2^\xi}{1-\rho}
+
\frac{2}{e(1-\xi)|\log\rho|}
\right]k^{-\xi}\\
&=
C_{\rho,\xi}k^{-\xi}.    
\end{align*}
This proves \eqref{eq:geometric-convolution}.
\end{proof}

\begin{lemma}
\label{lem:scaled-harmonic-sum}
For \(r>0\) and \(K\ge1\), define
\(
H_r(K):=\sum_{k=1}^{K}k^{-r}.
\)
Then,
\[
H_r(K)
\le
\begin{cases}
1+\dfrac{K^{1-r}-1}{1-r},
&0<r<1,\\[2mm]
1+\log K,
&r=1,\\[1mm]
1+\dfrac{1}{r-1},
&r>1.
\end{cases}
\]
Consequently, for every fixed \(r>0\), there exists a constant
\(C_r>0\), independent of \(K\), such that
\[
H_r(K)
\le
C_r(1+\log K)K^{[1-r]_+}.
\]
More generally, for every \(b\in\mathbb R\),
\[
K^{-b}H_r(K)
\le
C_r(1+\log K)
K^{-\min\{b,b+r-1\}}.
\]
\end{lemma}

\begin{proof}
If \(0<r<1\), the integral comparison gives
\[
H_r(K)
\le
1+\int_1^K x^{-r}\,dx
=
1+\frac{K^{1-r}-1}{1-r}.
\]
If \(r=1\), then
\[
H_1(K)
\le
1+\int_1^K\frac{1}{x}\,dx
=
1+\log K.
\]
Finally, if \(r>1\),
\[
H_r(K)
\le
1+\int_1^\infty x^{-r}\,dx
=
1+\frac{1}{r-1}.
\]
These three estimates imply
\[
H_r(K)
\le
C_r(1+\log K)K^{[1-r]_+}
\]
for some constant \(C_r>0\).

Multiplying both sides by \(K^{-b}\) gives
\[
K^{-b}H_r(K)
\le
C_r(1+\log K)
K^{-b+[1-r]_+}.
\]
Since
\[
b-[1-r]_+
=
\begin{cases}
b+r-1,&0<r<1,\\
b,&r\ge1,
\end{cases}
=
\min\{b,b+r-1\},
\]
we obtain
\[
K^{-b}H_r(K)
\le
C_r(1+\log K)
K^{-\min\{b,b+r-1\}}.
\]
\end{proof}

\section{Complete Proofs for Section~\ref{sec:distributed-optimization}}
\label{app:section-iv-proofs}

This appendix supplies the complete proofs of the results used in the
finite-time analysis of the distributed oracle and in the outer-search
guarantee. Throughout, constants may depend on the fixed game, compact
sets, communication graph, and algorithm parameters, but not on the
queried tuple~$w$, the iteration budget~$K$, or the initial primal
points in~$\mathcal Z$.

\subsection{Proof of Proposition~\ref{prop:distributed-equivalence}}

Given any feasible point \((\theta,t)\) of
\eqref{eq:tuple-inner-epigraph}, for \( i\in\mathcal I,\) set
\(
z_i=(\theta,t).
\)
The resulting collection satisfies all the local and consensus
constraints of the consensus formulation, and its objective value is
\(
(N)^{-1}\sum_{i=1}^N t_i=t.
\)
Conversely, consider any feasible collection
\(\{z_i\}_{i\in\mathcal I}\) of the consensus formulation.
Because \(\mathcal G\) is connected, the consensus constraints imply
that
\(
z_1=\cdots=z_N=(\theta,t)
\)
for some \((\theta,t)\). The local constraints then imply
\(h_i^w(\theta)\leq t\) for every \(i\in\mathcal I\).
Therefore, \((\theta,t)\) is feasible for
\eqref{eq:tuple-inner-epigraph}, and the two corresponding points
have the same objective value.

\subsection{Proof of Lemma~\ref{lem:dual-iterate-bound}}
\label{app:proof-dual-iterates}

Let
\[
 C_\theta:=\max_{\theta\in\Theta}\|\theta\|_2,
 \qquad
 C_t:=\max\{|t_{\min}|,|t_{\max}|\},
\]
and recall the uniform bounds
\[
 G_\psi:=\max_{i\in\mathcal I,\,y\in\mathcal Y}
 \|\Delta\psi(y;i)\|_2,
\]
and
\[
 U:=\max_{i\in\mathcal I,\,y\in\mathcal Y}
 |\Delta u_i(y;i)|.
\]
These constants are finite by compactness and continuity. Since
$|\ell_i(\theta,y)|\le C_\theta G_\psi+U$, compactness of
$\mathcal T=[t_{\min},t_{\max}]$ gives
\[
 \bigl|h_i^w(\theta_i[k])-t_i[k]\bigr|
 \le B_\mu,
 \qquad
 B_\mu:=C_\theta G_\psi+U+C_t,
\]
uniformly over $w$, $i$, and $k$.

The dual update and $\eta_k\beta_k\le c\beta_0\le1$ yield
\begin{align*}
\mu_i[k+1]
&=\Bigl[(1-\eta_k\beta_k)\mu_i[k]+\eta_k\bigl(h_i^w(\theta_i[k])-t_i[k]\bigr)\Bigr]_+\\
&\le (1-\eta_k\beta_k)\mu_i[k]+\eta_kB_\mu.
\end{align*}
We prove by induction that $\mu_i[k]\le B_\mu/\beta_k$.
The claim holds at $k=0$ because $\mu_i[0]=0$. If it holds at
iteration $k$, then
\begin{align*}
\mu_i[k+1]
&\le (1-\eta_k\beta_k)\frac{B_\mu}{\beta_k}
     +\eta_kB_\mu\\
&=\frac{B_\mu}{\beta_k}
 \le\frac{B_\mu}{\beta_{k+1}},
\end{align*}
where the last inequality follows from $\beta_{k+1}\le\beta_k$.
Therefore,
\[
 0\le\mu_i[k]\le\frac{B_\mu}{\beta_k}
 =\frac{B_\mu}{\beta_0}(k+1)^q
\]
for all $i$, $k$, and $w$.

For later use, Algorithm~\ref{alg:drpdm} and $\|s_i[k]\|_2\le G_\psi$ also give
\begin{align}
 \|g_i[k]\|_2
 &\le \mu_i[k](G_\psi+1)+\frac1N
 \le C_g(k+1)^q,                                      \label{eq:app-g-bound}\\
 |r_i[k]|
 &\le B_\mu+\beta_k\mu_i[k]
 \le 2B_\mu,                                          \label{eq:app-r-bound}
\end{align}
where
\[
 C_g:=\frac{B_\mu}{\beta_0}(G_\psi+1)+\frac1N.
\]
All these bounds are uniform over $w$.

\subsection{Proof of Lemma~\ref{lem:consensus-error}}
\label{app:proof-consensus}

Let
\[
 P:=I_N-\sigma L_{\mathcal G},
 \qquad
 J:=\frac1N\mathbf1\mathbf1^\top.
\]
Because the graph is undirected,
$L_{\mathcal G}\mathbf1=0$ and
$\mathbf1^\top L_{\mathcal G}=0^\top$, so
\[
 P\mathbf1=\mathbf1,
 \qquad
 \mathbf1^\top P=\mathbf1^\top.
\]
Moreover,
$d_i=e_i^\top L_{\mathcal G}e_i\le\lambda_{\max}(L_{\mathcal G})$.
The condition $\sigma\le1/\lambda_{\max}(L_{\mathcal G})$ therefore
implies
\[
 P_{ii}=1-\sigma d_i\ge0,
 \qquad
 P_{ij}=\sigma[A_{\mathcal G}]_{ij}\ge0
 \quad(i\ne j).
\]
Hence, $P$ is nonnegative and doubly stochastic.

Let the eigenvalues of $L_{\mathcal G}$ be ordered as
\[
0=\lambda_1(L_{\mathcal G})
 <\lambda_2(L_{\mathcal G})
 \le\cdots\le\lambda_N(L_{\mathcal G}).
\]
Connectivity gives $\lambda_2(L_{\mathcal G})>0$. Since
$0<\sigma\le1/\lambda_{\max}(L_{\mathcal G})$, spectral decomposition
gives
\begin{equation}
 \rho_{\mathcal G}:=\|P-J\|_2
 =1-\sigma\lambda_2(L_{\mathcal G})<1.                 \label{eq:app-rho}
\end{equation}

Let $Z[k]$ and $\mathsf G[k]$ have rows $z_i[k]^\top$ and
$g_i[k]^\top$, respectively. Since
\[
 \sum_{j\in\mathcal N_i}[A_{\mathcal G}]_{ij}
 (z_j[k]-z_i[k])
 =-(L_{\mathcal G}Z[k])_i^\top,
\]
the stacked unprojected primal update is
\[
 PZ[k]-\eta_k\mathsf G[k].
\]
Also,
\[
 e_k=\frac1{\sqrt N}\|(I_N-J)Z[k]\|_F.
\]
The pairwise-variance identity and nonexpansiveness of the common
row-wise projection onto $\mathcal Z$ give
\begin{align*}
e_{k+1}
&\le\frac1{\sqrt N}
 \bigl\|(I_N-J)(PZ[k]-\eta_k\mathsf G[k])\bigr\|_F\\
&\le\frac1{\sqrt N}\|P-J\|_2\|(I_N-J)Z[k]\|_F\\
&\quad+\frac{\eta_k}{\sqrt N}
 \|I_N-J\|_2\|\mathsf G[k]\|_F\\
&\le\rho_{\mathcal G}e_k+\eta_kC_g(k+1)^q,
\end{align*}
where the last inequality uses~\eqref{eq:app-g-bound}.
Substituting $\eta_k=c(k+1)^{-\tau}$ yields
\begin{equation}
 e_{k+1}\le\rho_{\mathcal G}e_k
 +cC_g(k+1)^{-(\tau-q)}.                              \label{eq:app-cons-recursion}
\end{equation}

Put $\xi:=\tau-q>0$. Iterating~\eqref{eq:app-cons-recursion} gives,
for every $k\ge1$,
\[
 e_k\le\rho_{\mathcal G}^ke_0
 +cC_g\sum_{s=0}^{k-1}
 \rho_{\mathcal G}^{k-1-s}(s+1)^{-\xi}.
\]
Let
\[
 D_{\mathcal Z}:=\max_{z,z'\in\mathcal Z}\|z-z'\|_2.
\]
Convexity of $\mathcal Z$ implies $z_{\rm av}[0]\in\mathcal Z$ and
hence $e_0\le D_{\mathcal Z}$. If $0<\rho_{\mathcal G}<1$,
Lemma~14 gives
\[
 \rho_{\mathcal G}^k\le C_{\rho,1}k^{-\xi},
 \qquad
 \sum_{s=0}^{k-1}\rho_{\mathcal G}^{k-1-s}(s+1)^{-\xi}
 \le C_{\rho,2}k^{-\xi}.
\]
Thus, for a constant $C_e>0$ independent of $w$ and $k$,
\[
 e_k\le C_ek^{q-\tau}.
\]
When $\rho_{\mathcal G}=0$, the same conclusion follows directly
from~\eqref{eq:app-cons-recursion}. Enlarging $C_e$ to cover $k=0$
gives
\begin{equation}
 e_k\le C_e(k+1)^{q-\tau},
 \qquad k\ge0.                                        \label{eq:app-point-consensus}
\end{equation}

We next pass to the weighted ergodic averages. Since
\[
 \bar z_w[K]=\frac1{S_K}\sum_{k=0}^{K-1}\eta_kz_{\rm av}[k],
\]
for every player $i$,
\[
 \bar z_i[K]-\bar z_w[K]
 =\frac1{S_K}\sum_{k=0}^{K-1}\eta_k
 \bigl(z_i[k]-z_{\rm av}[k]\bigr).
\]
The definition of $e_k$ implies
$\|z_i[k]-z_{\rm av}[k]\|_2\le\sqrt N e_k$, and hence
\begin{equation}
 \max_i\|\bar z_i[K]-\bar z_w[K]\|_2
 \le\frac{\sqrt N}{S_K}\sum_{k=0}^{K-1}\eta_ke_k
 =D_K.                                                  \label{eq:app-ergodic-consensus}
\end{equation}
Furthermore,~\eqref{eq:app-point-consensus} gives
\begin{align*}
\sum_{k=0}^{K-1}\eta_ke_k
&\le cD_{\mathcal Z}
 +cC_e\sum_{k=1}^{K}k^{-(2\tau-q)}\\
&\le c\bigl[D_{\mathcal Z}+C_eH_{2\tau-q}(K)\bigr].
\end{align*}
Since
$S_K=c\sum_{k=1}^{K}k^{-\tau}\ge cK^{1-\tau}$,
\[
 D_K\le\sqrt N K^{\tau-1}
 \bigl[D_{\mathcal Z}+C_eH_{2\tau-q}(K)\bigr].
\]
Applying Lemma~\ref{lem:scaled-harmonic-sum} in Appendix~\ref{app:auxiliary-results} with $b=1-\tau$ and $r=2\tau-q$ yields
\[
 D_K\le C_D(1+\log K)
 K^{-\min\{\tau-q,1-\tau\}},
\]
which proves both claims.

\subsection{Proof of Lemma~\ref{lem:ergodic-primal-dual}}
\label{app:proof-ergodic-pd}

For compactness, write
\[
 c_i[k]:=h_i^w(\theta_i[k])-t_i[k]
\]
and define the local unregularized Lagrangian
\[
 \mathcal L_i(z_i,\mu_i)
 :=\frac{t_i}{N}+\mu_i\bigl(h_i^w(\theta_i)-t_i\bigr).
\]
Thus,
\[
 \mathcal L(Z,\mu):=\sum_{i=1}^N\mathcal L_i(z_i,\mu_i).
\]
Fix a consensus comparator $Z_c=\mathbf1z^\top$, with
$z\in\mathcal Z$, and any $\mu\in\mathbb R_+^N$. Let
\[
 v_i[k]:=\sum_{j=1}^NP_{ij}z_j[k].
\]
Then the primal update is
\[
 z_i[k+1]=P_{\mathcal Z}(v_i[k]-\eta_kg_i[k]).
\]
Because $z=P_{\mathcal Z}(z)$, nonexpansiveness gives
\begin{multline*}
\|z_i[k+1]-z\|_2^2
\le\|v_i[k]-\eta_kg_i[k]-z\|_2^2\\
=\|v_i[k]-z\|_2^2
 -2\eta_k\langle g_i[k],z_i[k]-z\rangle\\
+2\eta_k\langle g_i[k],z_i[k]-v_i[k]\rangle
 +\eta_k^2\|g_i[k]\|_2^2.
\end{multline*}
Since $P$ is doubly stochastic, Jensen's inequality yields
\[
 \sum_i\|v_i[k]-z\|_2^2
 \le\sum_i\|z_i[k]-z\|_2^2.
\]
Moreover, $\mu_i[k]\ge0$ makes
$\mathcal L_i(\cdot,\mu_i[k])$ convex, and
$g_i[k]\in\partial_z\mathcal L_i(z_i[k],\mu_i[k])$. Therefore,
\[
 \langle g_i[k],z_i[k]-z\rangle
 \ge\mathcal L_i(z_i[k],\mu_i[k])
 -\mathcal L_i(z,\mu_i[k]).
\]
Because $P=I_N-\sigma L_{\mathcal G}$,
\[
 \sum_i\langle g_i[k],z_i[k]-v_i[k]\rangle
 =\sigma\langle\mathsf G[k],L_{\mathcal G}Z[k]\rangle_F.
\]
Summing the preceding projection inequality over $i$ gives
\begin{align}
&\eta_k\sum_i\bigl[
 \mathcal L_i(z_i[k],\mu_i[k])
 -\mathcal L_i(z,\mu_i[k])\bigr]\notag\\
&\le\frac12\sum_i\bigl(
 \|z_i[k]-z\|_2^2-\|z_i[k+1]-z\|_2^2\bigr)\notag\\
&\quad+\eta_k\sigma
 \langle\mathsf G[k],L_{\mathcal G}Z[k]\rangle_F
 +\frac{\eta_k^2}{2}\sum_i\|g_i[k]\|_2^2.             \label{eq:app-primal-one-step}
\end{align}

For the dual update, projection onto $\mathbb R_+$ gives, for every
comparator $\mu_i\ge0$,
\begin{align*}
|\mu_i[k+1]-\mu_i|^2
&\le\bigl|\mu_i[k]-\mu_i+\eta_kr_i[k]\bigr|^2\\
&=|\mu_i[k]-\mu_i|^2
 +2\eta_k(\mu_i[k]-\mu_i)c_i[k]\\
&\quad-2\eta_k\beta_k\mu_i[k](\mu_i[k]-\mu_i)\\
&\quad+\eta_k^2|r_i[k]|^2.
\end{align*}
Completing the square gives
\[
 (\mu_i-\mu_i[k])\mu_i[k]
 =\frac{\mu_i^2}{4}
 -\left(\mu_i[k]-\frac{\mu_i}{2}\right)^2
 \le\frac{\mu_i^2}{4}.
\]
Since $\mathcal L_i(z_i[k],\cdot)$ is affine,
\[
 \mathcal L_i(z_i[k],\mu_i)
 -\mathcal L_i(z_i[k],\mu_i[k])
 =(\mu_i-\mu_i[k])c_i[k].
\]
Rearranging the dual projection inequality therefore gives
\begin{align}
&\eta_k\bigl[
 \mathcal L_i(z_i[k],\mu_i)
 -\mathcal L_i(z_i[k],\mu_i[k])\bigr]\notag\\
&\le\frac12\bigl(
 |\mu_i[k]-\mu_i|^2-|\mu_i[k+1]-\mu_i|^2\bigr)\notag\\
&\quad+\frac{\eta_k^2}{2}|r_i[k]|^2
 +\frac{\eta_k\beta_k}{4}\mu_i^2.                    \label{eq:app-dual-one-step}
\end{align}

Let $\bar Z[K]$ have rows $\bar z_i[K]^\top$ and define
\[
 \bar\mu[K]:=\frac1{S_K}\sum_{k=0}^{K-1}\eta_k\mu[k].
\]
Convexity of $\mathcal L(\cdot,\mu)$ and affinity of
$\mathcal L(Z_c,\cdot)$ imply
\begin{multline*}
S_K\bigl[
 \mathcal L(\bar Z[K],\mu)-\mathcal L(Z_c,\bar\mu[K])
 \bigr]\\
\le\sum_{k=0}^{K-1}\eta_k
 \bigl[
 \mathcal L(Z[k],\mu)-\mathcal L(Z_c,\mu[k])
 \bigr].
\end{multline*}
Adding~\eqref{eq:app-dual-one-step}, summed over $i$, to
\eqref{eq:app-primal-one-step}, and then summing over $k$, makes all
distance terms telescope. Since $z_i[0],z\in\mathcal Z$,
$\mu_i[0]=0$, and the dropped terminal squared distances are
nonnegative, we obtain
\begin{align}
&S_K\bigl[
 \mathcal L(\bar Z[K],\mu)-\mathcal L(Z_c,\bar\mu[K])
 \bigr]\notag\\
&\le\frac{ND_{\mathcal Z}^2}{2}
 +\sigma\sum_{k=0}^{K-1}\eta_k
 \bigl|\langle\mathsf G[k],L_{\mathcal G}Z[k]\rangle_F\bigr|\notag\\
&\quad+\frac12\sum_{k=0}^{K-1}\eta_k^2
 \sum_i\bigl(\|g_i[k]\|_2^2+|r_i[k]|^2\bigr)\notag\\
&\quad+\left(
 \frac12+\frac14\sum_{k=0}^{K-1}\eta_k\beta_k
 \right)\|\mu\|_2^2.                                  \label{eq:app-telescoped}
\end{align}

It remains to control the mixing and squared-direction terms. Since
$L_{\mathcal G}J=0$, Cauchy--Schwarz,
\eqref{eq:app-g-bound}, and the definition of $e_k$ yield
\begin{align*}
\bigl|\langle\mathsf G[k],L_{\mathcal G}Z[k]\rangle_F\bigr|
&\le\|\mathsf G[k]\|_F\|L_{\mathcal G}\|_2
 \|(I_N-J)Z[k]\|_F\\
&\le N\|L_{\mathcal G}\|_2C_g(k+1)^qe_k.
\end{align*}
Furthermore,~\eqref{eq:app-g-bound}--\eqref{eq:app-r-bound} imply
\[
 \sum_i\bigl(\|g_i[k]\|_2^2+|r_i[k]|^2\bigr)
 \le C_1(k+1)^{2q}
\]
for a uniform constant $C_1>0$. Substituting these two bounds into
\eqref{eq:app-telescoped}, absorbing all fixed constants into
$C_0>0$, and dividing by $S_K$ gives
\[
 \mathcal L(\bar Z[K],\mu)-\mathcal L(Z_c,\bar\mu[K])
 \le R_K+Q_K\|\mu\|_2^2,
\]
where $R_K$ and $Q_K$ are exactly those defined in Lemma~\ref{lem:ergodic-primal-dual}. The
constant $C_0$ is independent of $w$, $K$, and the initial primal
points because $D_{\mathcal Z}$ and all constants above are uniform.

\subsection{Proof of Theorem~\ref{thm:drpdm-convergence}}
\label{app:proof-inner-oracle}

Fix $w\in\mathcal W$ and choose
\[
\begin{aligned}
 \theta_w^\star&\in\arg\min_{\theta\in\Theta}F(\theta,w),\\
 z_w^\star&:=(\theta_w^\star,G(w)),\\
 Z_c&:=\mathbf1(z_w^\star)^\top.
\end{aligned}
\]
The point $z_w^\star$ is feasible for the epigraph problem, so
\begin{equation}
 \mathcal L(Z_c,\bar\mu[K])\le G(w).                    \label{eq:app-comparator-feasible}
\end{equation}
Let $i_K$ attain
\[
 F(\bar\theta_w[K],w)
 =h_{i_K}^w(\bar\theta_w[K]),
\]
let $e_{i_K}$ be the corresponding coordinate unit vector, and set
$\bar t_w[K]:=N^{-1}\sum_i\bar t_i[K]$. Then
\[
 \mathcal L(\bar Z[K],e_{i_K})
 =\bar t_w[K]+h_{i_K}^w(\bar\theta_{i_K}[K])
 -\bar t_{i_K}[K].
\]
Since $h_i^w$ is $G_\psi$-Lipschitz and
$\|\bar z_{i_K}[K]-\bar z_w[K]\|_2\le D_K$,
\begin{align*}
F(\bar\theta_w[K],w)
&\le h_{i_K}^w(\bar\theta_{i_K}[K])+G_\psi D_K\\
&\le\mathcal L(\bar Z[K],e_{i_K})+(G_\psi+1)D_K.
\end{align*}
Applying Lemma~\ref{lem:ergodic-primal-dual} with $\mu=e_{i_K}$, using
$\|e_{i_K}\|_2=1$, and invoking
\eqref{eq:app-comparator-feasible} gives
\begin{equation}
0\le\widehat G_K(w)-G(w)
\le R_K+Q_K+(G_\psi+1)D_K.                             \label{eq:app-oracle-decomp}
\end{equation}
The lower bound follows directly from
$G(w)=\min_{\theta\in\Theta}F(\theta,w)$.

It remains to bound the three terms uniformly. Let
$H_r(K):=\sum_{k=1}^Kk^{-r}$. From the definitions of $R_K$ and
$Q_K$, the pointwise consensus estimate
\eqref{eq:app-point-consensus}, and $S_K\ge cK^{1-\tau}$,
the right-hand side of~\eqref{eq:app-oracle-decomp} is bounded by a
uniform constant multiple of
\[
K^{\tau-1}\bigl[
1+H_{2\tau-2q}(K)+H_{\tau+q}(K)+H_{2\tau-q}(K)
\bigr].
\]
Indeed,
\begin{align*}
\sum_{k=0}^{K-1}\eta_k^2(k+1)^{2q}
&=c^2H_{2\tau-2q}(K),\\
\sum_{k=0}^{K-1}\eta_k(k+1)^qe_k
&\le cC_eH_{2\tau-2q}(K),\\
\sum_{k=0}^{K-1}\eta_k\beta_k
&=c\beta_0H_{\tau+q}(K),
\end{align*}
while Lemma~\ref{lem:consensus-error} supplies the $H_{2\tau-q}(K)$ term through $D_K$.

Applying Lemma~\ref{lem:scaled-harmonic-sum} gives the respective decay exponents
\[
\begin{aligned}
&1-\tau, &&\min\{1-\tau,\tau-2q\},\\
&\min\{1-\tau,q\},
&&\min\{1-\tau,\tau-q\},
\end{aligned}
\]
up to a factor $1+\log K$. Their minimum is
\[
 \gamma_c(\tau,q)=\min\{q,\tau-2q,1-\tau\}>0.
\]
Therefore, for a constant $C>0$ independent of $w$, $K$, and the
initial primal points,
\[
0\le\widehat G_K(w)-G(w)
\le C(1+\log K)K^{-\gamma_c(\tau,q)}.
\]
Setting \(\tau=3/4\) and \(q=1/4\) gives
\[
\gamma_c(\tau,q)
=
\min\{q,\tau-2q,1-\tau\}
=
\frac14.
\]
Therefore,
\[
0\leq \widehat G_K(w)-G(w)
\leq
C(1+\log K)K^{-1/4}
=
\widetilde{\mathcal O}(K^{-1/4})
\]
uniformly over \(w\in\mathcal W\), which proves
Theorem~\ref{thm:drpdm-convergence}.

\subsection{Proof of Lemma~\ref{lem:outer-lipschitz}}
\label{app:proof-outer-regularity}

Let $w=(y_1,\ldots,y_{d+1})$ and
$\widetilde w=(\widetilde y_1,\ldots,\widetilde y_{d+1})$.
For any $i$, $\theta$, and $y,\widetilde y\in\mathcal Y$, the
definitions of $\ell_i$, $C_\theta$, $L_\psi$, and $L_u$ give
\begin{align*}
|\ell_i(\theta,y)-\ell_i(\theta,\widetilde y)|
&\le |\langle\theta,
 \Delta\psi(y;i)-\Delta\psi(\widetilde y;i)\rangle|\\
&\quad+|\Delta u_i(y;i)-\Delta u_i(\widetilde y;i)|\\
&\le2(C_\theta L_\psi+L_u)\|y-\widetilde y\|_2\\
&=L_{\rm game}\|y-\widetilde y\|_2.
\end{align*}
Taking finite maxima over players and tuple components preserves
this common Lipschitz bound. Thus, for every $\theta\in\Theta$,
\[
 |F(\theta,w)-F(\theta,\widetilde w)|
 \le L_{\rm game}d_{\mathcal W}(w,\widetilde w).
\]
Let $\theta_{\widetilde w}$ minimize $F(\cdot,\widetilde w)$. Then
\begin{align*}
G(w)-G(\widetilde w)
&\le F(\theta_{\widetilde w},w)
 -F(\theta_{\widetilde w},\widetilde w)\\
&\le L_{\rm game}d_{\mathcal W}(w,\widetilde w).
\end{align*}
Interchanging $w$ and $\widetilde w$ proves
\[
 |G(w)-G(\widetilde w)|
 \le L_{\rm game}d_{\mathcal W}(w,\widetilde w).
\]
If $w^\star$ and $w$ lie in the same cell $\mathcal W_{h,r}$,
the partition-diameter condition consequently gives
\[
G(w^\star)-G(w)
\le L_{\rm game}\alpha_{\mathcal W}\rho^h
=\nu_{\rm game}\rho^h.
\]

\subsection{Proof of Theorem~\ref{thm:exotic-global}}
\label{app:proof-global-outer}

Theorem~\ref{thm:drpdm-convergence} gives the uniform one-sided
oracle bound for all \(w\in\mathcal W,\, K\geq1,\)
\begin{equation}
0
\leq
\widehat G_K(w)-G(w)
\leq
C(1+\log K)K^{-1/4},
\label{eq:app-oracle-rate}
\end{equation}
for some constant \(C>0\) independent of \(w\) and \(K\).

Lemma~\ref{lem:outer-lipschitz} and the geometric
partition-diameter condition in~\eqref{eq:partition-diameter}
imply that every cell \(\mathcal W_{h,r}\) containing an optimizer
\(w^\star\) satisfies
\[
\sup_{w\in\mathcal W_{h,r}}
\bigl(G(w^\star)-G(w)\bigr)
\leq
\nu_{\mathrm{game}}\rho^h.
\]
Moreover, \eqref{eq:near-optimality-dimension} supplies the admissible
near-optimality exponent \(d_G\).

To make the oracle error in~\eqref{eq:app-oracle-rate} comparable
to the depth-\(h\) resolution \(\rho^h\), an inner budget of order
\[
K_h
=
\widetilde{\mathcal O}(\rho^{-4h})
\]
is sufficient. At the same depth, the number of near-optimal cells
is at most of order \(\rho^{-d_Gh}\). Consequently, the total inner
work required by the EXOTIC search to resolve depth \(h\) is
\[
\widetilde{\mathcal O}
\left(
\rho^{-(d_G+4)h}
\right).
\]
Equivalently, under total inner-iteration budget \(n\), the finest
achievable resolution satisfies
\[
\rho^h
=
\widetilde{\mathcal O}
\left(
n^{-1/(d_G+4)}
\right).
\]
This is the budget argument in
\cite[Theorem~3.3 and Corollary~3.5]{EXOTIC}, with the logarithmic
oracle factor absorbed into \(\widetilde{\mathcal O}\). Hence, for
the refined estimate returned by EXOTIC,
\begin{equation}
p^\star-\widehat G_n
=
\widetilde{\mathcal O}
\left(
n^{-1/(d_G+4)}
\right),
\label{eq:app-exotic-estimate}
\end{equation}
where
\(
\widehat G_n
:=
\widehat G_{K_{\mathrm{final}}}(\widehat w_n).
\)

Under the EXOTIC budget rule,
\[
h_{\max}
=
\left\lfloor
\frac{2n}{5B(1+\log n)^2}
\right\rfloor,
\qquad
K_{\mathrm{final}}
=
\lfloor h_{\max}/2\rfloor
=
\widetilde{\Theta}(n).
\]
Therefore, the final re-evaluation satisfies
\[
0
\leq
\widehat G_n-G(\widehat w_n)
=
\widetilde{\mathcal O}(n^{-1/4}).
\]
Combining this bound with~\eqref{eq:app-exotic-estimate} gives
\begin{align*}
p^\star-G(\widehat w_n)
&=
p^\star-\widehat G_n
+\widehat G_n-G(\widehat w_n)\\
&=
\widetilde{\mathcal O}
\left(
n^{-1/(d_G+4)}
\right)
+
\widetilde{\mathcal O}(n^{-1/4}).
\end{align*}
Since \(d_G\geq0\),
\[
\frac{1}{d_G+4}\leq\frac14,
\]
so the final re-evaluation error is no slower than the outer-search
rate. Thus,
\[
0
\leq
p^\star-G(\widehat w_n)
=
\widetilde{\mathcal O}
\left(
n^{-1/(d_G+4)}
\right).
\]

\section{The EXOTIC Tree-Search Procedure}
\label{app:exotic-algorithm}

For completeness, Algorithm~\ref{alg:exotic-original} reproduces
the EXOTIC procedure of~\cite{EXOTIC}, with minor notational
changes and corrections of typographical errors. We use \(B\) for
the branching factor, reserving \(K\) for the DRPDM iteration
budget. For \(w\in\mathcal W\), \(z\in\varphi(w)\), and \(H\ge1\),
define
\[
\mathsf{Eval}_H(w,z)
:=
F\bigl(\operatorname{OPT}(w,z,H),w\bigr).
\]
In our implementation, \(\operatorname{OPT}\) is instantiated by
the distributed oracle in Algorithm~\ref{alg:drpdm}.

\begin{algorithm}[t]
\caption{EXOTIC tree search~\cite{EXOTIC}}
\label{alg:exotic-original}
\footnotesize
\begin{algorithmic}[1]
\Require Depth \(h_{\max}\), branching factor \(B\), representatives
\(\{(w_{h,r},z_{h,r})\}_{h,r}\)
\State Initialize \(\mathcal T^{(B)}\) with root
\((w_{0,1},G_{0,1}=0,z_{0,1},\#_{0,1}=0)\)
\For{\(r=1,\ldots,B\)}
    \State \(G_{1,r}\gets
    \mathsf{Eval}_{h_{\max}}(w_{1,r},z_{1,r})\)
    \State \(\#_{1,r}\gets h_{\max}\); append \((1,r)\) to the root
\EndFor
\Statex \textbf{Tree-expansion phase}
\For{\(h=1,\ldots,h_{\max}\)}
    \For{\(m=1,\ldots,\lfloor h_{\max}/h\rfloor\)}
        \State \(H\gets\lfloor h_{\max}/(hm)\rfloor\)
        \State \(
        (h,\bar r)\gets
        \displaystyle\arg\max_{\substack{
        (h,r)\in\mathcal T_{\rm leaf}^{(B)}\\
        \#_{h,r}\ge H}}
        G_{h,r}
        \)
        \State Append the \(B\) children of \((h,\bar r)\)
        \For{\(r=1,\ldots,B\)}
            \State \(G_{h+1,r}\gets
            \mathsf{Eval}_{H}(w_{h+1,r},z_{h+1,r})\)
            \State \(\#_{h+1,r}\gets H\)
        \EndFor
    \EndFor
\EndFor
\Statex \textbf{Re-evaluation phase}
\For{\(p=0,\ldots,\lfloor\log_2h_{\max}\rfloor\)}
    \State \(
    (h_p,r_p)\gets
    \displaystyle\arg\max_{\substack{
    (h,r)\in\mathcal T^{(B)}\\
    \#_{h,r}\ge2^p}}
    G_{h,r}
    \)
    \State \(G_{h_p,r_p}\gets
    \mathsf{Eval}_{\lfloor h_{\max}/2\rfloor}
    (w_{h_p,r_p},z_{h_p,r_p})\)
\EndFor
\State \(
\widehat p\gets
\displaystyle\arg\max_{
0\le p\le\lfloor\log_2h_{\max}\rfloor}
G_{h_p,r_p}
\)
\State \(\widehat w_n\gets
w_{h_{\widehat p},r_{\widehat p}}\)
\State \(\widehat G_n\gets
G_{h_{\widehat p},r_{\widehat p}}\)
\State \textbf{return} \(\widehat G_n,\widehat w_n\)
\end{algorithmic}
\end{algorithm}

\section{Projected Zeroth-Order Ascent}
\label{app:zo-algorithm}

Algorithm~\ref{alg:zeroth-order} gives the complete projected
zeroth-order procedure. Each outer iteration uses two DRPDM
evaluations, and the final tuple is re-evaluated using the larger
budget \(K_{\rm final}\). The exponents \(r_\delta\) and
\(r_\alpha\) control the decay of the perturbation radius and
outer stepsize, respectively; their numerical values are specified
in Section~\ref{sec:numerical-experiments}.

\begin{algorithm}[t]
\caption{Projected Zeroth-Order Ascent with DRPDM Evaluations}
\label{alg:zeroth-order}
\begin{algorithmic}[1]
\Require Outer budget \(T\); inner budgets
\(\{K_t\}_{t=0}^{T-1}\); final budget \(K_{\rm final}\);
\(\delta_0,\alpha_0>0\) with
\(\mathcal W_{\delta_0}\neq\varnothing\);
\(w_0\in\mathcal W_{\delta_0}\);
set
\(\delta_t=\delta_0(t+1)^{-r_\delta}\) and
\(\alpha_t=\alpha_0(t+1)^{-r_\alpha}\).
\For{$t=0,\ldots,T-1$}
  \State Sample $v_t\sim\operatorname{Unif}(\mathbb S^{D_{\mathcal W}-1})$
  \State $w_t^+\gets w_t+\delta_t v_t$
  \State Run Algorithm~\ref{alg:drpdm} on $w_t$ and $w_t^+$,
  each for $K_t$ iterations with multipliers initialized to zero
  \State Obtain $\widehat G_t^0$ and $\widehat G_t^+$
  \State $\widehat g_t\gets
    (D_{\mathcal W}/\delta_t)(\widehat G_t^+-\widehat G_t^0)v_t$
  \State $w_{t+1}\gets P_{\mathcal W_{\delta_{t+1}}}
    (w_t+\alpha_t\widehat g_t)$
\EndFor
\State Run Algorithm~\ref{alg:drpdm} on $w_T$ for
$K_{\mathrm{final}}$ iterations
\State \Return $w_T$, $\widehat\theta_T:=\bar\theta_{w_T}[K_{\mathrm{final}}]$,
and $\widehat G_T:=F(\widehat\theta_T,w_T)$
\end{algorithmic}
\end{algorithm}

\section{Exact Full-Domain Error for Linear--Quadratic Network Games}
\label{app:lq-full-domain-error}

This appendix derives the closed-form expression used to evaluate
the full-domain error in the linear--quadratic network-game
experiments. Recall the utility functions and quadratic potential
parameterization introduced in
Section~\ref{subsec:lq-experiments}. For a unilateral deviation
\(y=(a,a_i')\), the corresponding error is
\begin{equation}
\label{eq:lq-loss-coefficients}
\ell_i(\theta,y)
=
(a_i'-a_i)
\left[
A_i(\theta)(a_i'+a_i)
+B_i(\theta)
+\sum_{j\ne i}C_{ij}(\theta)a_j
\right],
\end{equation}
where 
$A_i(\theta)
:=
p_i+\frac12,$ $B_i(\theta)
:=
q_i-\zeta_i,$ and for all $j\ne i,$ $C_{ij}(\theta)
:=
r_{\min\{i,j\},\max\{i,j\}}
-\gamma V_{ij}.$
Define
\begin{equation}
\label{eq:lq-K}
K_i(\theta)
:=
|B_i(\theta)|
+
\sum_{j\ne i}|C_{ij}(\theta)|.
\end{equation}
Maximizing~\eqref{eq:lq-loss-coefficients} over the full action
domain gives
\begin{equation}
\label{eq:lq-exact-full-domain}
\mathcal E(\theta)
=
\max_{i\in[N]}H_i(\theta),
\end{equation}
where
\begin{equation}
\label{eq:lq-Hi}
H_i(\theta)
=
\begin{cases}
2K_i(\theta),
&
K_i(\theta)\geq2|A_i(\theta)|,\\[1mm]
|A_i(\theta)|+K_i(\theta)
+\dfrac{K_i(\theta)^2}{4|A_i(\theta)|},
&
0\leq K_i(\theta)<2|A_i(\theta)|.
\end{cases}
\end{equation}
When \(A_i(\theta)=0\), the first case applies. Hence, all reported
full-domain errors are evaluated exactly, without discretizing or
sampling the action space.

\end{document}